\documentclass[aps,pra,superscriptaddress,twocolumn,longbibliography]{revtex4-1}
\usepackage{graphicx, color, graphpap}
\usepackage{enumitem}
\usepackage{amssymb}
\usepackage{amsthm}
\usepackage{amsmath}
\usepackage{dsfont}
\usepackage{mathtools}
\usepackage{soul}

\usepackage[dvipsnames]{xcolor}
\usepackage{multirow}
\usepackage[colorlinks=true,citecolor=blue,linkcolor=magenta]{hyperref}
\usepackage[T1]{fontenc}

\usepackage{thmtools,thm-restate}
\usepackage{verbatim}
\usepackage{titlesec}

\usepackage{graphicx}

\usepackage{bbm}

\usepackage[ruled,vlined]{algorithm2e}

\newtheorem{theorem}{Theorem}
\newtheorem{lemma}{Lemma}

\newtheorem{proposition}{Proposition}

\newcommand{\bra}[1]{\langle#1|}

\newcommand{\ket}[1]{|#1\rangle}

\DeclareMathOperator*{\argmax}{argmax}

\begin{document}

\title{Limitations of Amplitude Encoding on Quantum Classification}

\author{Xin Wang}
\affiliation{
	Department of Automation, Tsinghua University, Beijing, 100084, P. R. China
}

\author{Yabo Wang}
\affiliation{
	State Key Laboratory of Mathematical Sciences, Academy of Mathematics and Systems Science, Chinese Academy of Sciences, Beijing 100190, China}
\affiliation{
	University of Chinese Academy of Sciences, Beijing 100049, P. R. China
}

\author{Bo Qi}
\email{qibo@amss.ac.cn}
\affiliation{
	State Key Laboratory of Mathematical Sciences, Academy of Mathematics and Systems Science, Chinese Academy of Sciences, Beijing 100190, China}
\affiliation{
	University of Chinese Academy of Sciences, Beijing 100049, P. R. China
}

\author{Rebing Wu}
\affiliation{
	Department of Automation, Tsinghua University, Beijing, 100084, P. R. China
}

\begin{abstract}
It remains unclear whether quantum machine learning  (QML) has real advantages when dealing with practical and meaningful tasks. Encoding classical data into quantum states is one of the key steps in QML. Amplitude encoding has been widely used owing to its remarkable efficiency in encoding a number of $2^{n}$ classical data into $n$ qubits simultaneously. However, the theoretical impact of amplitude encoding on QML has not been thoroughly investigated.  In this work we prove that under some broad and typical data assumptions, the average of  encoded quantum states via amplitude encoding tends to concentrate towards a specific state. This concentration phenomenon severely constrains the capability of quantum classifiers as it leads  to a loss barrier phenomenon, namely, the loss function has a lower bound that cannot be improved by any optimization algorithm. In addition, via numerical simulations, we reveal a counterintuitive phenomenon of amplitude encoding: as the amount of training data increases, the training error may increase rather than decrease, leading to reduced  decrease in prediction accuracy on new data. Our results highlight the limitations of amplitude encoding in QML and indicate that more efforts should be devoted to finding more efficient encoding strategies to unlock the full potential  of QML.
\end{abstract}

\maketitle
\section{Introduction}

Quantum machine learning (QML) has been expected to become a new technology that can surpass its classical counterpart due to quantum advantages in representation, parallelism and entanglement~\cite{biamonte2017quantum,xiao2023practical,azam2024optically,zheng2023efficient}.
Classification is a standard classical machine learning task that has been extensively studied~\cite{lecun1998gradient,krizhevsky2012imagenet,chen2024robust,erik2024consistent,vinod2024online,wang2024quantum}. Quantum classifiers~\cite{schuld2020circuit,perez2020data,henderson2020quanvolutional,hubregtsen2022training,huang2021power,xu2024quantum,mitsuda2024approximate,wang2024supervised,liu2021rigorous,jerbi2023quantum} have attracted increasing attention and claimed advantages on some toy classical datasets~\cite{schuld2020circuit,perez2020data,henderson2020quanvolutional,hubregtsen2022training,huang2021power,xu2024quantum,mitsuda2024approximate,wang2024supervised} and carefully designed datasets~\cite{liu2021rigorous,jerbi2023quantum}.
Although it is widely believed that quantum classifiers outperform classical machine learning when processing quantum data~\cite{cong2019quantum,huang2021information,huang2022quantum}, for practical classification tasks, it remains unclear whether QML has real advantage,
particularly in the current noisy intermediate-scale quantum era~\cite{schuld2022quantum,cerezo2022challenges,lau2022nisq}.

Quantum machine learning models~\cite{schuld2021supervised,caro2022generalization,huang2021power} are constructed using parameterized quantum circuits (PQCs), which consist of quantum gates with tunable parameters~\cite{benedetti2019parameterized}. There have been many works investigating the limitations of QML models. Most of them focused on the trainability issues caused by the phenomena of barren plateaus~\cite{mcclean2018barren,marrero2021entanglement,thanasilp2023subtleties,ragone2024lie} and having exponentially many spurious local minimum~\cite{you2021exponentially,anschuetz2022quantum}. Various kinds of strategies, such as, appropriate parameter initialization~\cite{zhang2022escaping,wang2024trainability}, symmetry-preserving ansatzes~\cite{meyer2023exploiting}, tree tensor networks~\cite{huggins2019towards}, convolutional neural networks~\cite{cong2019quantum}, local encoded cost function~\cite{cerezo2021cost},  and over-parameterization~\cite{larocca2023theory} have been proposed to mitigate the training issue. Few works have investigated the generalization ability of QML~\cite{caro2022generalization,caro2021encoding,huang2021information}, and it was proved in Ref.~\cite{caro2022generalization} that good generalizations can be guaranteed from few data.  We note that good generalization alone does not necessarily guarantee good prediction, since the training error may be large and become the dominant term in the prediction error.

Encoding classical data into quantum states is a necessary and crucial step in QML~\cite{cerezo2022challenges,wiebe2020key,lloyd2020quantum,weigold2021expanding,rath2024quantum}.
There are mainly two encoding paradigms~\cite{jerbi2023quantum}:~(1) variational-encoding, where the data encoding and variational PQC are separated;~(2) data re-uploading, where  the data encoding and variational PQC are interleaved~\cite{perez2020data}. In this work we focus on the variational-encoding paradigm.  There are three primary quantum encoding methods: basis encoding, angle encoding and amplitude encoding~\cite{schuld2018supervised}. Basis encoding can only encode binary information by mapping each bit to either $\ket{0}$ or $\ket{1}$. Angle encoding embeds real numbers by mapping them to the angles of quantum rotational gates. The angle encoding is  easy to implement but not efficient. Amplitude encoding encodes complex numbers by mapping them to the probability amplitudes of quantum states. Despite having significant challenges in implementing amplitude encoding in practice, it can encode a number of $2^{n}$ data into $n$ qubits simultaneously.  In view of this remarkable efficiency, amplitude encoding has been widely employed to demonstrate quantum superiority on classical datasets~\cite{larose2020robust,grant2018hierarchical,hur2022quantum,mitsuda2024approximate,wang2024supervised,li2024ensemble}. However, the numerical simulations therein are mainly based on  the simple MNIST dataset~\cite{lecun1998mnist}, which may mask the limitations of QML caused by amplitude encoding.

There are only few works concerning the relationship between the power or limitations of QML and encoding strategies. 
Ref.~\cite{caro2021encoding} and Ref.~\cite{schuld2021effect} 
have respectively investigated the impact of encoding strategies on the generalization error and expressive power of QML.
For the limitations of QML, it was proved in Ref.~\cite{li2022concentration} that under reasonable assumptions, the average encoded quantum states via angle encoding concentrates to the maximally mixed state at an exponential rate as the depth of the encoding circuit increases. For quantum kernel methods which employ angle encoding to generate kernels~\cite{hubregtsen2022training,schuld2019quantum,wang2021towards,thanasilp2024exponential}, it was demonstrated that for a wide range of situations, values of quantum kernels over different input data concentrate towards some fixed value exponentially as the number of qubits increases~\cite{thanasilp2024exponential}. The concentration phenomenon severely limits the power of QML. In the case of amplitude encoding, most studies have primarily focused on two aspects: how to prepare the target encoded states~\cite{giovannetti2008quantum,park2019circuit,jiang2019experimental,plesch2011quantum, long2001efficient, araujo2021divide,gonzalez2024efficient}, and how to efficiently implement approximate amplitude encoding~\cite{nakaji2022approximate, mitsuda2024approximate,daimon2024quantum}.

In this work, we investigate the limitation of amplitude encoding on quantum classification. First, we  prove that under certain conditions, a loss barrier phenomenon occurs, that is, the loss function has a lower bound that cannot be improved by any optimization algorithm. Then, we theoretically show that under reasonable data assumptions, the amplitude encoding satisfies the conditions for the loss barrier phenomenon. Finally, we validate our findings via numerical simulations that the amplitude encoding is prone to suffer from the loss barrier phenomenon for sophisticated and commonly encountered datasets.
In machine learning a general belief is that an increase in the amount of training data generally benefits the performance. However, our numerical simulations reveal that under amplitude encoding as the amount of training data increases, although the generalization error decreases, the training error increases counterintuitively, resulting in overall poor prediction performance. This is essentially owing to the concentration phenomenon caused by the amplitude encoding. Our findings help understand the limitations of amplitude encoding in QML and highlight the need for finding more efficient encoding strategies to unlock the full power of QML.

The paper is organized as follows. In Sec. \ref{II}, we introduce the amplitude encoding, and set up the framework for quantum classification. We present the main results in Sec. \ref{III},  and conclude in Sec. \ref{IV}.

\section{Preliminary and Framework}\label{II}

\subsection{Amplitude encoding}
Without loss of generality, suppose that the classical data feature $\boldsymbol{x} \in \mathbb{R}^{2^{n}}$. By amplitude encoding, we can embed the feature $\boldsymbol{x}= \left[ x_0, x_1,\dots, x_{2^n-1} \right]^{\top}$ into a quantum state vector $\ket{\phi(\boldsymbol{x})}$ as
\begin{equation*}
\begin{aligned}
\ket{\phi(\boldsymbol{x})} = \frac{1}{C_{\boldsymbol{x}}} \sum_{i=0}^{2^{n}-1} x_i \ket{i},
\end{aligned}
\end{equation*}
where $C_{\boldsymbol{x}} = \sqrt{x_0^2 + x_1^2 + \cdots + x_{2^{n}-1}^2}$ is the normalization factor. An equivalent representation in terms of the density matrix  $\rho(\boldsymbol{x})$ reads
\begin{equation}
	\label{eq:amp}
\begin{aligned}
\rho(\boldsymbol{x}) &= \ket{\phi(\boldsymbol{x})} \bra{\phi(\boldsymbol{x})} \\
&= \frac{1}{C_{\boldsymbol{x}}^2} \sum_{i = 0}^{2^{n}-1} \sum_{j = 0}^{2^{n}-1 } x_{i} x_{j} \ket{i} \bra{j}.
\end{aligned}
\end{equation}

\subsection{Quantum classifier}

Consider a $K$-class classification problem. Denote the feature space by $\mathcal{X}=\mathbb{R}^{2^{n}}$ and the label space by $\mathcal{Y}=\{1,2,\dots,K\}$. Given a dataset $S=\left\{ \left( \boldsymbol{x}^{(m)}, y^{(m)} \right)  \right\} ^{M}_{m=1}$ containing $M$ independent and identically distributed ($i.i.d.$) samples, where each sample $\left( \boldsymbol{x}^{(m)}, {y}^{(m)} \right)  \in \mathcal{X} \times \mathcal{Y}$ obeys an unknown joint distribution $\mathcal{D}$, a quantum classifier learns from the dataset $S$ to obtain a hypothesis $h_S$. Our aim is to enable the hypothesis $h_S$ to accurately classify new, unseen samples that follow the same distribution $\mathcal{D}$.

To be specific, we embed the feature $\boldsymbol{x}^{(m)}$ into a feature quantum state $\rho\left( \boldsymbol{x}^{(m)} \right)$ through a PQC  $U(\boldsymbol{\theta})$. For classification, we  select $K$ fixed observables, $H_1,\cdots,H_K$, each corresponding to one of the $K$ classes. These observables are typically chosen to be the tensor product of Pauli operators, such as, ${Z \otimes Y, X \otimes Y,X \otimes Y \otimes Z,\cdots}$, where $X,Y,Z$ denote Pauli matrices. The output associated with the $k$-th class is defined as
\begin{equation*}
\begin{aligned}
h_k \left(\boldsymbol{x}^{(m)}, \boldsymbol{\theta} \right) = \operatorname{Tr} \left[H_k U(\boldsymbol{\theta}) \rho (\boldsymbol{x}^{(m)} ) U^{\dagger}(\boldsymbol{\theta})\right].
\end{aligned}
\end{equation*}
Let
\begin{equation*}
\begin{aligned}
\hat{y}^{(m)} = \argmax_k h_k(\boldsymbol{x}^{(m)},\boldsymbol{\theta}).
\end{aligned}
\end{equation*}
We now define the output of the learned hypothesis $h_S$ for the feature $\boldsymbol{x}^{(m)}$ to be $$h_S(\boldsymbol{x}^{(m)}) = \hat{y}^{(m)}.$$

\subsection{Learning framework}
To minimize the difference between the predicted label $\hat{y}^{(m)}$ of the hypothesis $h_S$ and the true label $y^{(m)}$ for the feature $\boldsymbol{x}^{(m)}$, we train the parameters of the quantum classifier on dataset $S = \{(\boldsymbol{x}^{(m)},y^{(m)}) \}_{m=1}^{M}$ using the cross-entropy loss with the softmax function:~\cite{li2022concentration}
\begin{equation}\label{eq:celoss}
\begin{aligned}
	&\mathcal{L}_S(\boldsymbol{\theta}) = \frac{1}{M} \sum_{m=1}^{M} \ell(\boldsymbol{\theta}; \boldsymbol{x}^{(m)},y^{(m)}), \\
	&\ell(\boldsymbol{\theta};\boldsymbol{x}^{(m)},y^{(m)}) = - \sum_{k=1}^{K} \boldsymbol{y}_k^{(m)} \ln \left( \frac{e^{h_k(\boldsymbol{x}^{(m)},\boldsymbol{\theta})}}{\sum_{j=1}^{K} e^{ h_j(\boldsymbol{x}^{(m)},\boldsymbol{\theta})} } \right).
\end{aligned}
\end{equation}
Here, the label $y^{(m)}$ is encoded as a one-hot vector $\boldsymbol{y}^{(m)}$ of dimension $K$, where only one element is 1 and all the other elements are 0. The $k$-th element of the one-hot vector $\boldsymbol{y}^{(m)}$  is denoted by $\boldsymbol{y}^{(m)}_k$, and $\boldsymbol{y}_k^{(m)}$ equals 1 if the feature $\boldsymbol{x}^{(m)}$ belongs to the $k$-th class.

To evaluate the  performance of a quantum classifier on dataset $S$, we define the training error of the hypothesis $h_S$ as
$$
\begin{aligned}
\widehat{R}_S(h_S) = \frac{1}{M} \sum_{m=1}^{M} \mathbbm{1} \left( h_S(\boldsymbol{x}^{(m)}) \neq y^{(m)} \right),
\end{aligned}
$$
where $\mathbbm{1}(\cdot)$ denotes the indicator function. Generally, a smaller value of the loss function corresponds to a lower training error.

Recall that the ultimate goal of the hypothesis $h_S$  is to classify new samples drawn from the distribution $\mathcal{D}$. We define its prediction error $R(h_S)$ as
$$
\begin{aligned}
	R(h_S) = \underset{(\boldsymbol{x},y) \sim \mathcal{D}}{\mathbb{E}}\left[ \mathbbm{1}\left( h_S(\boldsymbol{x}) \neq y \right)  \right].
\end{aligned}
$$
Since the true labels of new samples and the distribution $\mathcal{D}$ are unknown, the prediction error $R(h_S)$ cannot be obtained directly. To evaluate it, we decompose the prediction error $R(h_S)$ into the training error $\widehat{R}_S(h_S)$ and the generalization error, which is defined by
\begin{equation}
	\label{eq:gen}
\begin{aligned}
	\operatorname{gen}(h_S) = R(h_S) - \widehat{R}_S(h_S).
\end{aligned}
\end{equation}
In numerical simulations, a more intuitive metric for evaluating the classification performance is the accuracy. The training accuracy of hypothesis $h_S$  is defined as $$\widehat{A}_S(h_S) = 1 - \widehat{R}_S(h_S),$$
and the prediction accuracy is defined as $$A(h_S) = 1 - R(h_S).$$ 
It is clear that to ensure accurate predictions, quantum classifiers must exhibit both high training accuracy and low generalization error.

\section{Main Results}\label{III}

In this section we  present limitations of quantum classifiers caused by amplitude encoding. We start in Subsec.~\ref{subsec:loss_barrier} by establishing a lower bound of the loss function that holds true for any optimization algorithm, we  refer to this as the loss barrier phenomenon in our paper. Then in Subsec.~\ref{subsec:concentration} we demonstrate that under reasonable data assumptions, amplitude encoding leads to the loss barrier phenomenon.  We validate the limitations of amplitude encoding in Subsec.~\ref{sub:trainability} via numerical simulations. In  Subsec.~\ref{sub:real-world}, we further extend our simulations to real-world datasets, showing that the loss barrier phenomenon is widespread when using amplitude encoding.

\subsection{Loss barrier}
\label{subsec:loss_barrier}
Intuitively, to facilitate classification, the features of different classes should be as distinct from one another as possible. Poorly separated feature distributions make the classification task challenging. We show that  having similar expectations of quantum encoded states between different classes can lead to a loss barrier phenomenon, where the loss function exhibits an unfavorable lower bound.

To quantify the distinguishability between quantum states, we adopt the trace distance. For quantum states $\rho_1$ and $\rho_2$, their trace distance reads
\begin{equation}
\begin{aligned}
	T(\rho_1,\rho_2)=\frac{1}{2}\|\rho_1 - \rho_2\|_{1},
\end{aligned}
\end{equation}
where $\|\cdot\|_{1}$ is the Schatten-1 norm. It is clear that $T(\rho_1,\rho_2) \in [0,1]$.  Let $\rho(\boldsymbol{x})$ denote the encoded quantum state of feature $\boldsymbol{x}$. We say that the trace distance between the expectations of the encoded states of any two different classes is less than $\epsilon$ if, for all $k$ and $l$,
\begin{equation*}
\begin{aligned}
T \left( \underset{\boldsymbol{x} \sim  \mathcal{D}_{\mathcal{X}_k}}{\mathbb{E}}\left[\rho(\boldsymbol{x})\right] , \underset{\boldsymbol{x} \sim \mathcal{D}_{\mathcal{X}_l}}{\mathbb{E}}\left[\rho(\boldsymbol{x})\right]\right) \leqslant \epsilon, 
\end{aligned}
\end{equation*}
where $\mathcal{D}_{\mathcal{X}_k}$ and $\mathcal{D}_{\mathcal{X}_l}$ represent the distributions of features in classes $k$ and $l$, respectively. Under this condition, we demonstrate in Theorem~\ref{thm:loss_barrier} that the loss function exhibits an unfavorable lower bound.

\begin{theorem}
	\label{thm:loss_barrier}
	For a $K$-class classification, we employ the cross-entropy loss function $\mathcal{L}_S(\boldsymbol{\theta})$  defined in Eq.~\eqref{eq:celoss}. The quantum classifier is trained on a balanced training set $S = \{(\boldsymbol{x}^{(m)},y^{(m)})\}_{m=1}^{M}$, where  each class contains $M/K$ samples. Suppose the eigenvalues of each observable $H_k$ belong to $[-1,1]$, for $k=1,\ldots, K$. If the trace distance between the expectations of encoded states of any different classes is less than  $\epsilon$, then for any PQC $U(\boldsymbol{\theta})$ and optimization algorithm, we have
	\begin{equation}
		\label{eq:trace_between_class}
	\begin{aligned}
	\mathcal{L}_S(\boldsymbol{\theta}) \geqslant \ln\left[K - 4(K-1)\epsilon \right]
	\end{aligned}
	\end{equation}
	with probability at least $1 - 8e^{-M \epsilon^2 /8K}$.
\end{theorem}

Theorem~\ref{thm:loss_barrier}  clearly shows that if the expectations of encoded quantum states are similar across different classes,  then the cross-entropy loss function $\mathcal{L}_S(\boldsymbol{\theta})$ has a lower bound close to $\ln K$.  Notably, a loss value of  $\ln K$ implies that the quantum classifier is effectively classifying all samples entirely at random. Thus, in the case of Theorem~\ref{thm:loss_barrier}, the quantum classifier only has very poor prediction accuracy. 

The detailed proof of Theorem~\ref{thm:loss_barrier} is provided in  Appendix~\ref{app:proof-loss-barrier}. The basic idea is that if the amount of data in the training dataset is sufficiently large, then the average states of the encoded quantum states across different classes tend to become similar. This concentration phenomenon makes it difficult to distinguish different features. As a result, the loss function approaches to the completely random classification value, $\ln K$. 

We highlight that  Theorem~\ref{thm:loss_barrier}  fundamentally differs from Proposition 4 in Ref.~\cite{li2022concentration}. First, Proposition 4 in Ref.~\cite{li2022concentration} emphasized  the challenge of small gradients during the training process. However, it is important to note that small gradients do not necessarily preclude  successful training; rather, they indicate that more optimization iterations may be required. In contrast, our Theorem~\ref{thm:loss_barrier} reveals a fundamental limitation on the loss function:  regardless of the number of optimization iterations or the choice of optimization algorithm (whether gradient-based or not), the loss function has a lower bound close to  the value for random classification ($\ln K$), once the conditions of Theorem~\ref{thm:loss_barrier} are met. Second,  Proposition 4 assumed that the expectations of  encoded quantum states of different classes concentrate around the maximally mixed state,  while Theorem~\ref{thm:loss_barrier} imposes a less restrictive condition, requiring only that the trace distances between the expectations of encoded states across different classes are small. In addition, we note that the numerical simulations in Ref.~\cite{li2022concentration} support our Theorem~\ref{thm:loss_barrier}. Specifically,  Fig.~6.(c) and Fig.~7.(c) in Ref.~\cite{li2022concentration} demonstrated that for binary classification ($K=2$), the training losses do not decrease and stay around ln2. This observation aligns well with the prediction of Theorem~\ref{thm:loss_barrier}.

\subsection{Concentration caused by amplitude encoding}
\label{subsec:concentration}

In this subsection, we demonstrate three scenarios in which amplitude encoding results in the concentration of the expectations of encoded states across different classes.

First, we point out that the distinguishability between two quantum states $\rho_1$ and $\rho_2$ cannot be increased by measuring an observable $H$ after applying any PQC $U(\boldsymbol{\theta})$. This can be seen from the following inequality, which reads
\begin{equation}\label{eq:H_leq_1}
\begin{aligned}
	&\Big| \operatorname{Tr}\left[H U (\boldsymbol{\theta}) \rho_1 U^{\dagger}(\boldsymbol{\theta})\right] - \operatorname{Tr}\left[H U (\boldsymbol{\theta}) \rho_2 U^{\dagger}(\boldsymbol{\theta})\right]  \Big|  \\
	\leqslant &\left\| U^{\dagger}(\boldsymbol{\theta}) H U(\boldsymbol{\theta}) \right\|_{\infty} \cdot \left\| \rho_1 - \rho_2 \right\|_{1} \\
	\leqslant &2 T(\rho_1,\rho_2).
\end{aligned}
\end{equation}
Here, $\|\cdot\|_{\infty}$ denotes the Schatten-$\infty$ norm (also known as the spectral norm), and we have assumed the eigenvalues of the observable $H$ are bounded within the interval $[-1,1]$.  

Next, we introduce Lemma~\ref{lem:hoffding} dealing with the trace distance between the averaged and expected encoded states. 

For a given dataset $S = \{ (\boldsymbol{x}^{(m)},y^{(m)}) \}_{m=1}^{M}$ consisting of $M$ classical data, each $\boldsymbol{x}^{(m)}$ is drawn from a distribution $\mathcal{D}_{\mathcal{X}}$, and $y^{(m)}$ is its label. After amplitude encoding, the corresponding encoded state is denoted by $\rho(\boldsymbol{x}^{(m)})$. We denote by $\bar{\rho}_{M}$ the averaged encoded state, and $\mathbb{E}[{\rho}] = {\mathbb{E}}_{\boldsymbol{x} \sim \mathcal{D}_{\mathcal{X}}}\left[\rho(\boldsymbol{x})\right]$ expected encoded states. 

\begin{lemma}
	\label{lem:hoffding}
Given an arbitrary $\epsilon \in (0,1)$, we have
	$$
	\begin{aligned}
	 T \left( \overline{\rho}_M,\mathbb{E}[\rho] \right)  \leq \epsilon, 
	\end{aligned}
	$$
	with probability at least $1 - 4 e^{-M \epsilon^{2}/2}$. 
\end{lemma}
The detailed proof of Lemma~\ref{lem:hoffding}  is provided in Appendix~\ref{app:concentration-proof}. From Lemma~\ref{lem:hoffding}, it is clear that as the sample size $M$ increases, the averaged encoded state approaches to the expected encoded states.  Therefore, in the following numerical simulations, we utilize the averaged encoded state as a proxy for the expected encoded states. 

We now  illustrate the three scenarios in which the concentration phenomenon occurs after amplitude encoding.

\begin{proposition}
	\label{pro:min-max-con}
Assume that all elements in the feature $\boldsymbol{x} \in \mathbb{R}^{2^{n}}$ have the same sign,  and the elements satisfy $|x_i|\in[m, M]$. If $\left|\frac{m}{M} -1 \right| < \epsilon$, then after amplitude encoding, we have
	\begin{equation*}
	\begin{aligned}
		T \left( \rho(\boldsymbol{x}) ,\frac{1}{2^{n}} \mathbf{1} \mathbf{1}^{\top}  \right) \leqslant \sqrt{2 \epsilon},
	\end{aligned}
	\end{equation*}
where the state $\frac{1}{2^{n}} \boldsymbol{1} \boldsymbol{1}^{\top} = \ket{+}_{n} \bra{+}_{n} $, with $\ket{+}_n = H^{\otimes n} \ket{0}_{n} $ being the superposition of all computational  basis states.
\end{proposition}

Note that in Proposition~\ref{pro:min-max-con}, we only assume that all the feature elements have the same sign, and the variance of these elements is small, i.e., $|\frac{m}{M} -1| < \epsilon$. No assumptions have been made on  the mean value of each class. As illustrated in Fig.~\ref{fig:example_1}(a), consider two classes which are clearly distinguishable. It is straightforward to verify that the features of each class satisfy the conditions in Proposition~\ref{pro:min-max-con}. Then after amplitude encoding, their expected states concentrate to $\frac{1}{2^{n}} \boldsymbol{1} \boldsymbol{1}^{\top} = \ket{+}_{n} \bra{+}_{n} $, as illustrated in Fig.~\ref{fig:example_1}(b). Therefore, Proposition~\ref{pro:min-max-con} reveals a significant limitation of amplitude encoding: it may erase the mean value information in the classical features that is crucial for classification.

\begin{figure}[htpb]
	\centering
	\includegraphics[width=0.46\textwidth]{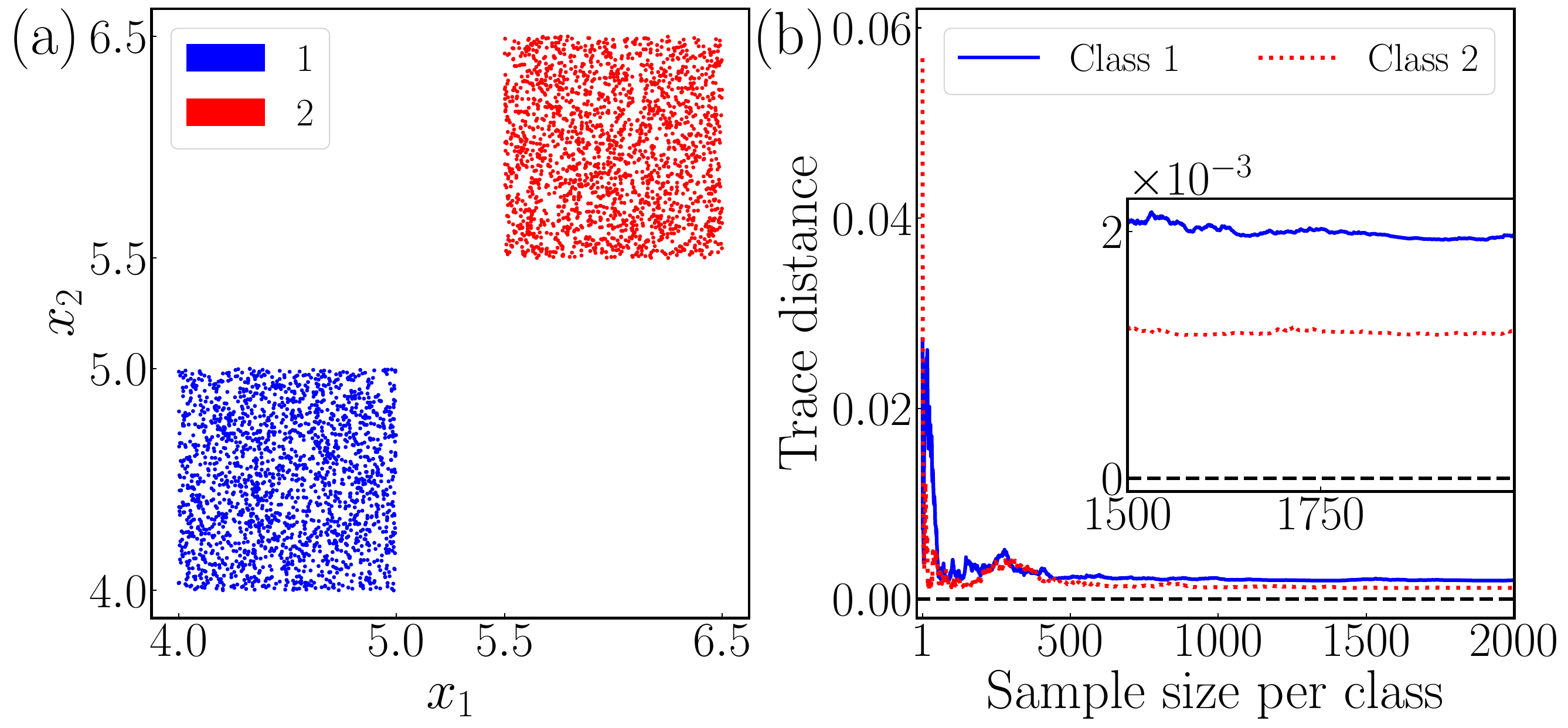}
	\caption{(a) A binary classification dataset. For class 1, features $x_1$ and $x_2$ obey the uniform distribution $\mathcal{U}[4,5]$, and for class 2, features $x_1$ and $x_2$ follow the uniform distribution $\mathcal{U}[5.5,6.5]$. Each class contains 2000 samples. (b) The trace distance between the averaged encoded states and the superposition state $\frac{1}{2} \mathbf{1} \mathbf{1}^{\top}$, where the solid blue line represents class 1 and the red dotted represents class 2. }
	\label{fig:example_1}
\end{figure}

\begin{proposition}
	\label{pro:max_mixed}
Denote by $\mathcal{D}_{\mathcal{X}}$ the distribution of the feature $\boldsymbol{x} \in \mathbb{R}^{2^{n}}$. Assume that the elements of $\boldsymbol{x}$ are i.i.d with an expected value of $0$. In addition, the distribution is symmetric, i.e., $p(x_j) = p(-x_j)$ for all $j \in [0:2^{n}-1]$. Then after amplitude encoding, the expectation of  encoded state is the maximally mixed state, i.e.,
	\begin{equation*}
	\begin{aligned}
	\underset{\boldsymbol{x} \sim  \mathcal{D}_{\mathcal{X}}}{\mathbb{E}}\left[\rho(\boldsymbol{x})\right] = \frac{\boldsymbol{I}}{2^{n}}.
	\end{aligned}
	\end{equation*}
\end{proposition}

Note that when the elements of the feature $\boldsymbol{x}$ are i.i.d., and symmetric about the mean value, to make the distribution satisfies the assumptions in Proposition~\ref{pro:max_mixed}, we can employ the standard  $z$-score normalization~\cite{fei2021z}, namely, letting $z = \frac{x - \mu}{\sigma}$, where $\mu$ is the mean of $x$ and $\sigma$ is the standard variance. 

We illustrate Proposition~\ref{pro:max_mixed} by using a binary classification dataset in Fig.~\ref{fig:example_2}. It is clear that the features of class 1 and class 2 can be distinguished by their distances to the origin (or the $\ell_2$ norm). However, after amplitude encoding, their averaged encoded states concentrate to the maximally mixed state. 
\begin{figure}[htpb]
	\centering
	\includegraphics[width=0.46\textwidth]{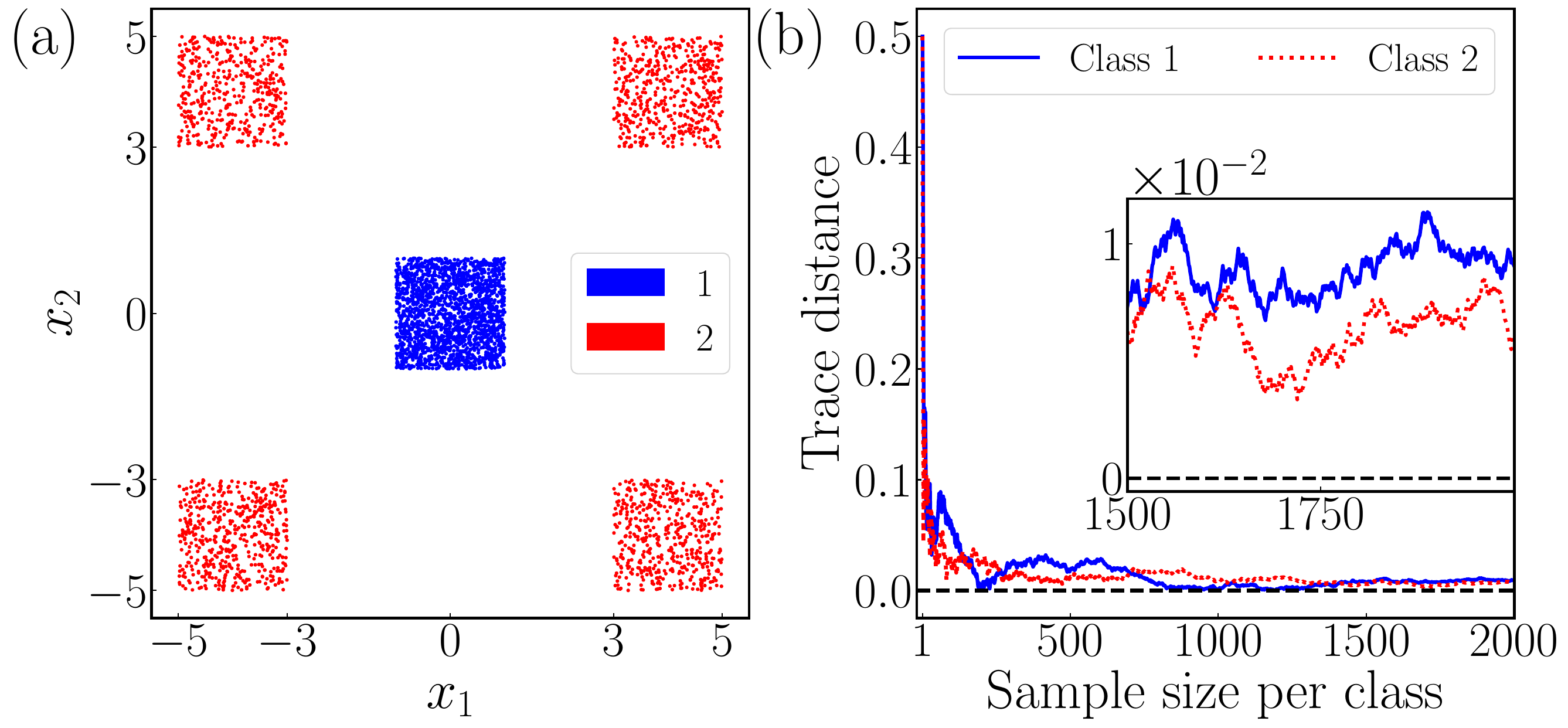}
	\caption{(a) A binary classification dataset. For class 1, features $x_1$ and $x_2$ obey the  uniform distribution $\mathcal{U}[-1,1]$. For class 2, features $x_1$ and $x_2$ follow the uniform distribution over the intervals $[-5,-3] \cup [3,5]$. Each class contains 2000 samples. (b) The trace distance between the averaged encoded state and the maximally mixed state $\frac{\boldsymbol{I}}{2}$.}
	\label{fig:example_2}
\end{figure}

\begin{proposition}
	\label{pro:binary_situation}
	For a binary classification dataset, denote by $\mathcal{D}_{\mathcal{X}_i}$ the distribution of feature $\boldsymbol{x}\in \mathbb{R}^{2^{n}}$ in class $i$, $i=1,2$.  If for any $j \in [0:2^{n}-1]$, the probability density functions $p_{i,j}(x)$ of the $j$-th element in class $i$ satisfy  $p_{1,j}(x) = -p_{2,j}(x)$, then after amplitude encoding, the expected states of the two classes are identical, i.e.,
	\begin{equation*}
		\begin{aligned}
		\underset{\boldsymbol{x} \sim \mathcal{D}_{\mathcal{X}_1}}{\mathbb{E}}\left[\rho(\boldsymbol{x})\right] = \underset{\boldsymbol{x} \sim \mathcal{D}_{\mathcal{X}_2}}{\mathbb{E}}\left[\rho(\boldsymbol{x})\right].
		\end{aligned}
	\end{equation*}
\end{proposition}

Proposition~\ref{pro:binary_situation} indicates that the concentration phenomenon caused by amplitude encoding is not limited to convergence towards a specific state. We exemplify Proposition~\ref{pro:binary_situation} in Fig.~\ref{fig:example_3} by showing an example where the features of class 1 and class 2 are linearly separable, but the trace distance of their averaged encoded states converges to zero as the the sample size increases.
\begin{figure}[htpb]
	\centering
	\includegraphics[width=0.46\textwidth]{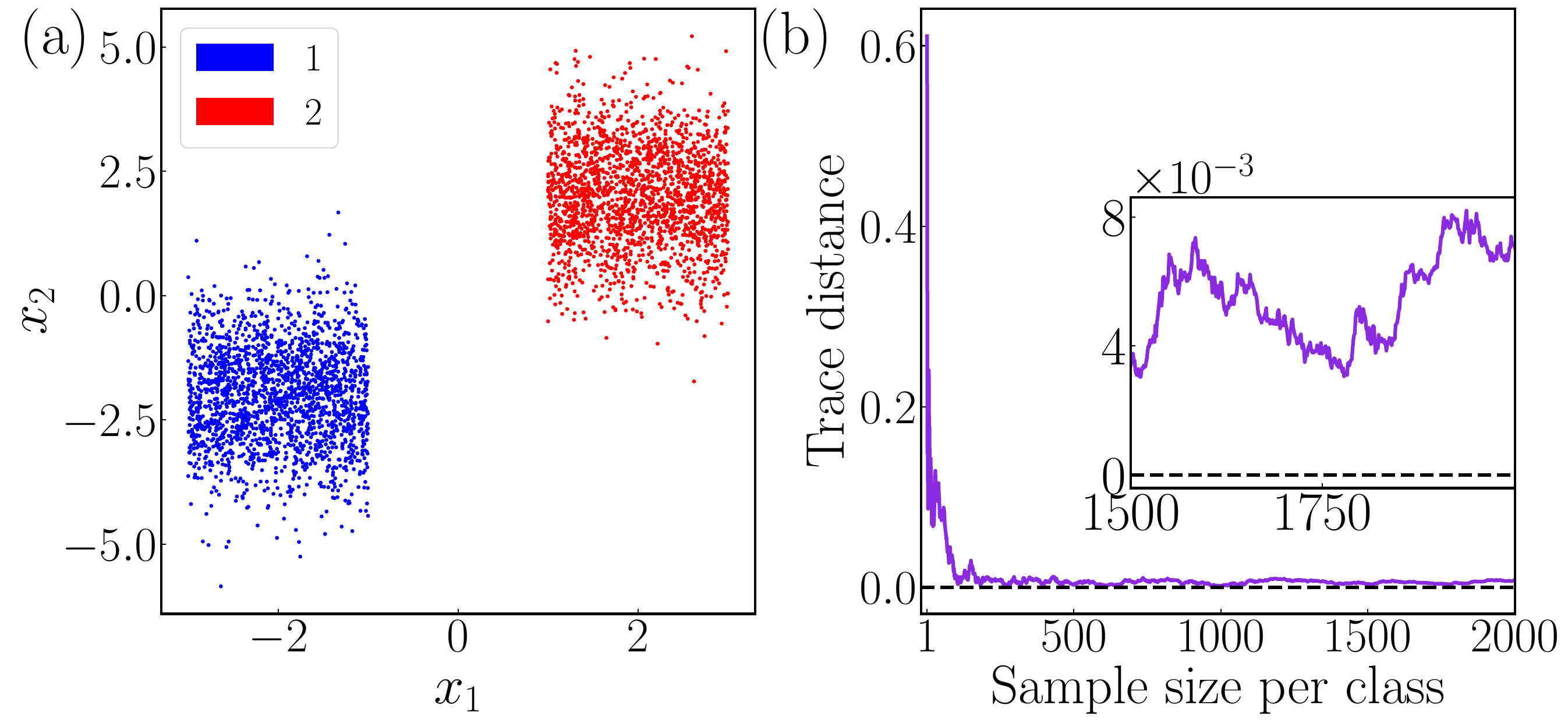}
	\caption{(a) A binary classification dataset. For class 1, $x_1$ follows the uniform distribution $\mathcal{U}[-3,-1]$, and $x_2$ follows the normal distribution $\mathcal{N}(-2,1)$. For class 2, $x_1$ obeys the uniform distribution $\mathcal{U}[1,3]$, and $x_2$ obeys the normal distribution $\mathcal{N}(2,1)$. Each class contains 2000 samples. (b) The trace distance between the averaged encoded states of class 1 and class 2.}
	\label{fig:example_3}
\end{figure}

The propositions in this subsection demonstrate that amplitude encoding may erase crucial information needed for classification, leading to the concentration phenomenon. This can result in poor performance of quantum classifiers. The detailed proofs of the propositions are provided in Appendix~\ref{app:concentration-proof}.

\subsection{Numerical validations}
\label{sub:trainability}

In this subsection, we verify the loss barrier phenomenon caused by amplitude encoding via numerical simulations. 

To demonstrate the limitations of amplitude encoding, we examine the binary datasets depicted in Fig.~\ref{fig:example_1}(a) to Fig.~\ref{fig:example_3}(a). Under amplitude encoding, if even these simple binary datasets cannot be effectively categorized, it is unlikely that more complex datasets could be successfully classified.  For these binary classifications tasks, we can just use a single qubit and employ a simple PQC, as illustrated in Fig.~\ref{fig:training_loss}(a). The classical data are first encoded into quantum states via amplitude encoding. Then the quantum states pass through a total of $L$ layers, each of which comprises three quantum gates: $Rz$, $Rx$, and $Rz$, with independent parameters. The initial parameters are sampled from the standard normal distribution $\mathcal{N}(0,1)$. We train the variational  parameters of the PQC by minimizing the loss function Eq.~\eqref{eq:celoss}. For each dataset, we choose the Pauli Z as the observable for class 1 and the Pauli X as the observable for class 2. We utilize the Adam optimizer with a learning rate of 0.01 for the training.  

We conduct numerical simulations for three scenarios with $L = 1$, $L=10$, and $L=30$, respectively. Each scenario is repeated 10 rounds, with random initializations of the circuit parameters and 1000 iterations per round. Fig~\ref{fig:training_loss}.(b) displays the training loss for the datasets from Fig.~\ref{fig:example_1}(a), Fig.~\ref{fig:example_2}(a), and Fig.~\ref{fig:example_3}(a) with $L = 10$. It is clear that at the initial stage of training, there is a significant descent in the loss function value, indicating that the gradient of the loss function does not vanish at the beginning of training. However, the loss function quickly converges to around $\ln2$, demonstrating the loss barrier phenomenon.  Similar results are observed for $L = 1$ and $L = 30$, and the details are presented in Appendix~\ref{app:trainability-validation}.

\begin{figure}[htpb]
	\centering
	\includegraphics[width=0.42\textwidth]{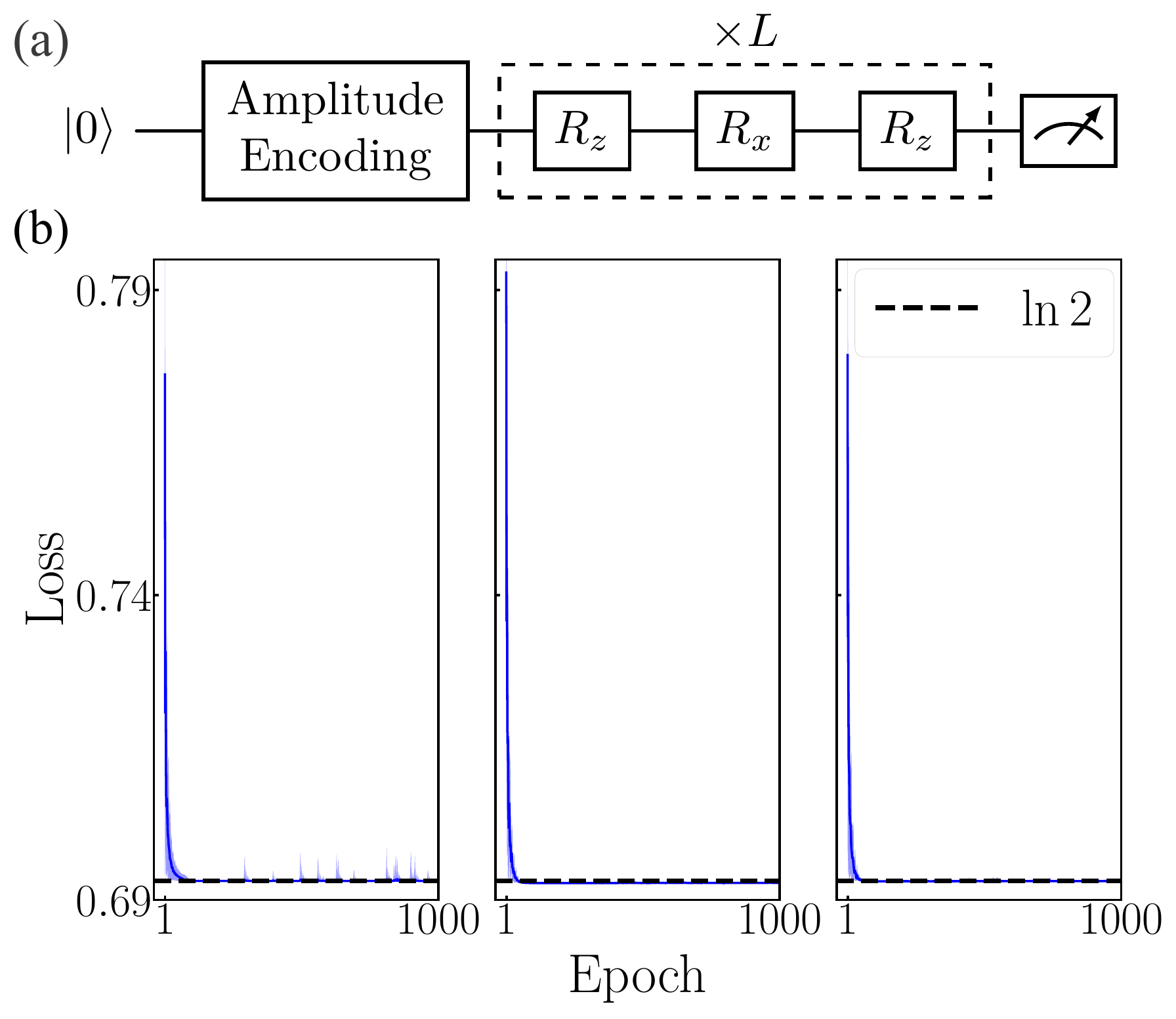}
	\caption{(a) The PQC with $L$ layers used for quantum binary classification. (b) Training loss for datasets from Fig.~\ref{fig:example_1}(a), Fig.~\ref{fig:example_2}(a), and Fig.~\ref{fig:example_3}(a) (from left to right). Blue solid lines depict the mean loss over 10 runs, and the shaded areas show their range. }
	\label{fig:training_loss}
\end{figure}

To further demonstrate the  negative impact of loss barrier, we present in  Table~\ref{tab:best_training_acc} the best training accuracy obtained within 1000 iterations for the three datasets. We find that the best training accuracy is around 0.5 with a relatively small standard deviation. This implies that, owing to the presence of loss barrier phenomenon, it is impossible to train a quantum classifier with satisfactory  performance. The essential reason is the concentration phenomenon induced by amplitude encoding.

\begin{table}[ht]
	\centering
	\caption{Best training accuracy for datasets from Fig.~\ref{fig:example_1}(a) to Fig.~\ref{fig:example_3}(a). The values in the table represent the mean $\pm$ standard deviation over 10 runs.}
	\begin{tabular}{c|c|c|c}
	  \toprule
	    &  {Fig.~\ref{fig:example_1}(a)}& {Fig.~\ref{fig:example_2}(a)} & Fig.~\ref{fig:example_3}(a) \\
		\hline
	 {Accuracy} & 0.534 $\pm$ 0.001 & 0.517  $\pm$ 0.000 & 0.511 $\pm$ 0.001 \\
	  \hline
	\end{tabular}
	\label{tab:best_training_acc}
  \end{table}



\subsection{Real-world datasets scenario}
\label{sub:real-world}

In this subsection, we consider real-world datasets and demonstrate that 
the concentration phenomenon induced by amplitude encoding is widespread.

We select common computer vision datasets for classification, including  forests and sea lakes from the EuroSAT dataset ($64 \times 64$ pixels)~\cite{helber2019eurosat}, adipose tissues and backgrounds from the PathMNIST dataset ($28 \times 28$ pixels)~\cite{yang2023medmnist}, and airplanes and birds from the CIFAR-10 dataset  ($32 \times 32$ pixels)~\cite{krizhevsky2009learning}. We label these pairs as class 1 and class 2, respectively, and compute the trace distance between the averaged encoded states of the two classes.  For comparison, we also consider the MNIST dataset and take digits 0 and 1 ($28 \times 28$ pixels) to be class 1 and class 2~\cite{lecun1998mnist}, which are frequently used in quantum classification tasks~\cite{larose2020robust,grant2018hierarchical,hur2022quantum,mitsuda2024approximate,wang2024supervised,li2024ensemble}. For all datasets, we resize the images to $32 \times 32$ pixels and employ 10 qubits for amplitude encoding.
\begin{figure}[htpb]
	\centering
	\includegraphics[width=0.4\textwidth]{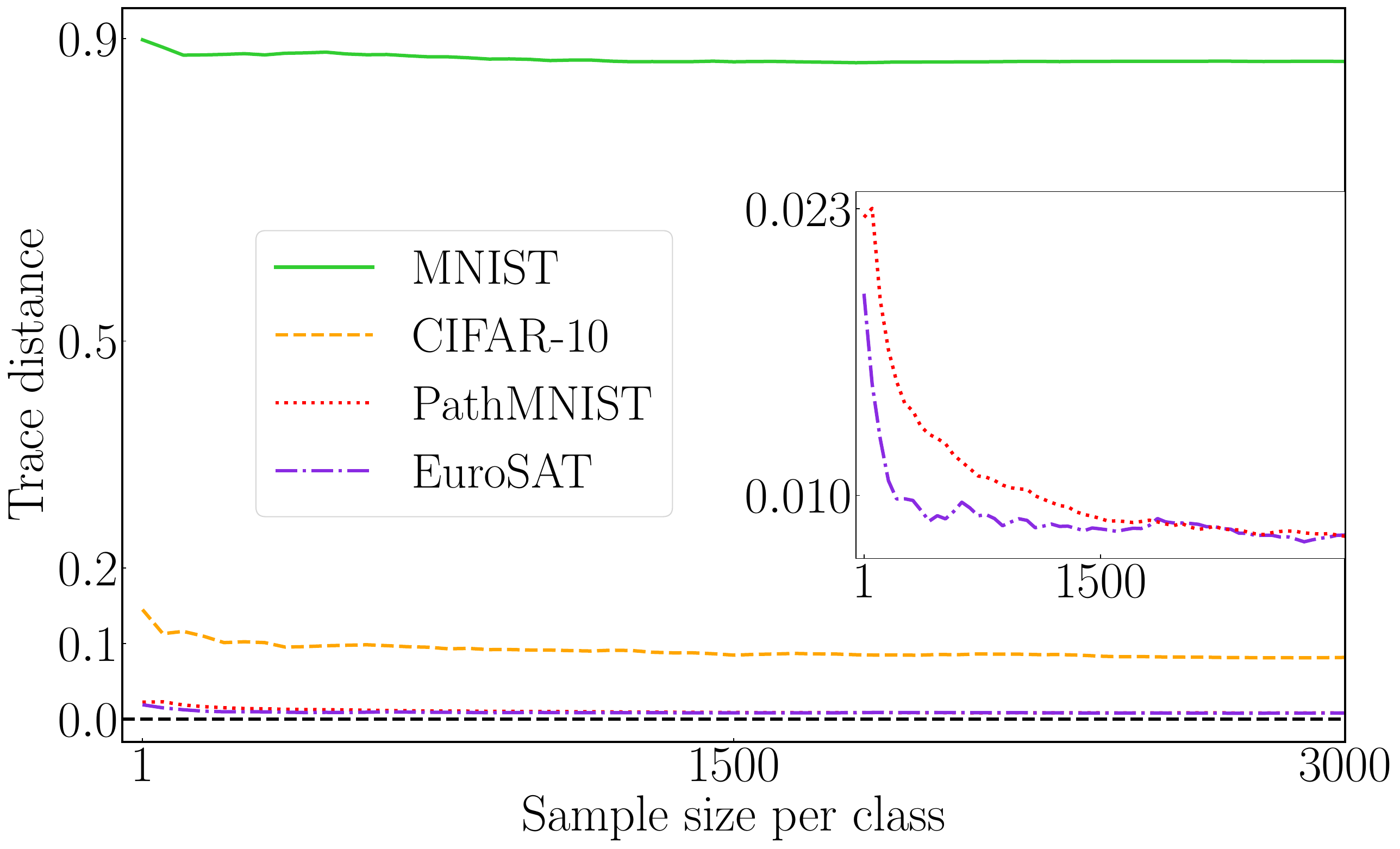}
	\caption{The trace distance between the averaged encoded states of class 1 and class 2 for different datasets.}
	\label{fig:real_datasets}
\end{figure}

\begin{figure*}[htpb]
	\centering
	\includegraphics[width=0.8\textwidth]{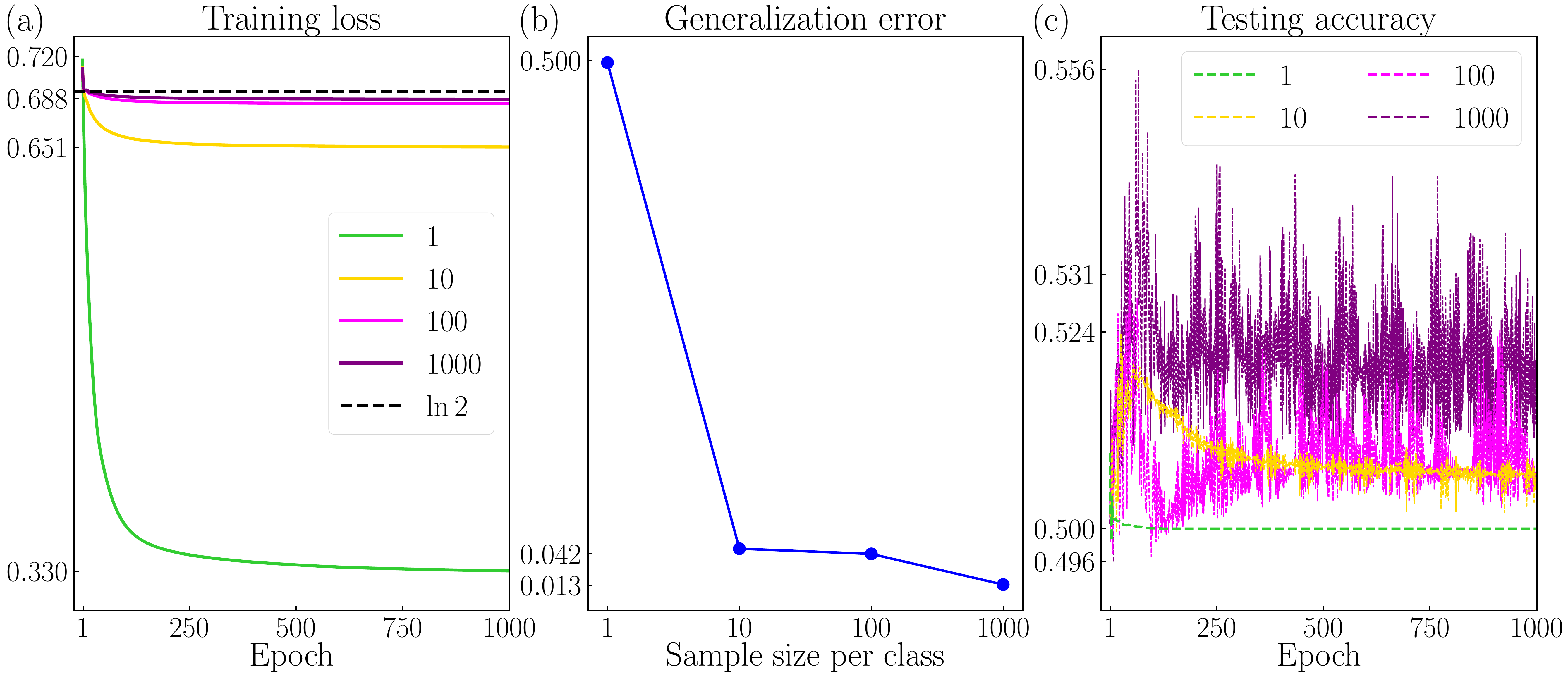}
	\caption{Quantum classification performance for EuroSAT dataset. (a) The average training loss of the quantum classifier over 10 runs for different sample sizes per class. (b) The average generalization error over 10 runs with 1000 iterations of each run (10000 results in total) versus different sample sizes per class. (c) The average testing accuracy obtained over 10 runs for  different sample sizes per class. }
	\label{fig:eurosat_loss}
\end{figure*}

\begin{figure*}[htpb]
	\centering
	\includegraphics[width=0.8\textwidth]{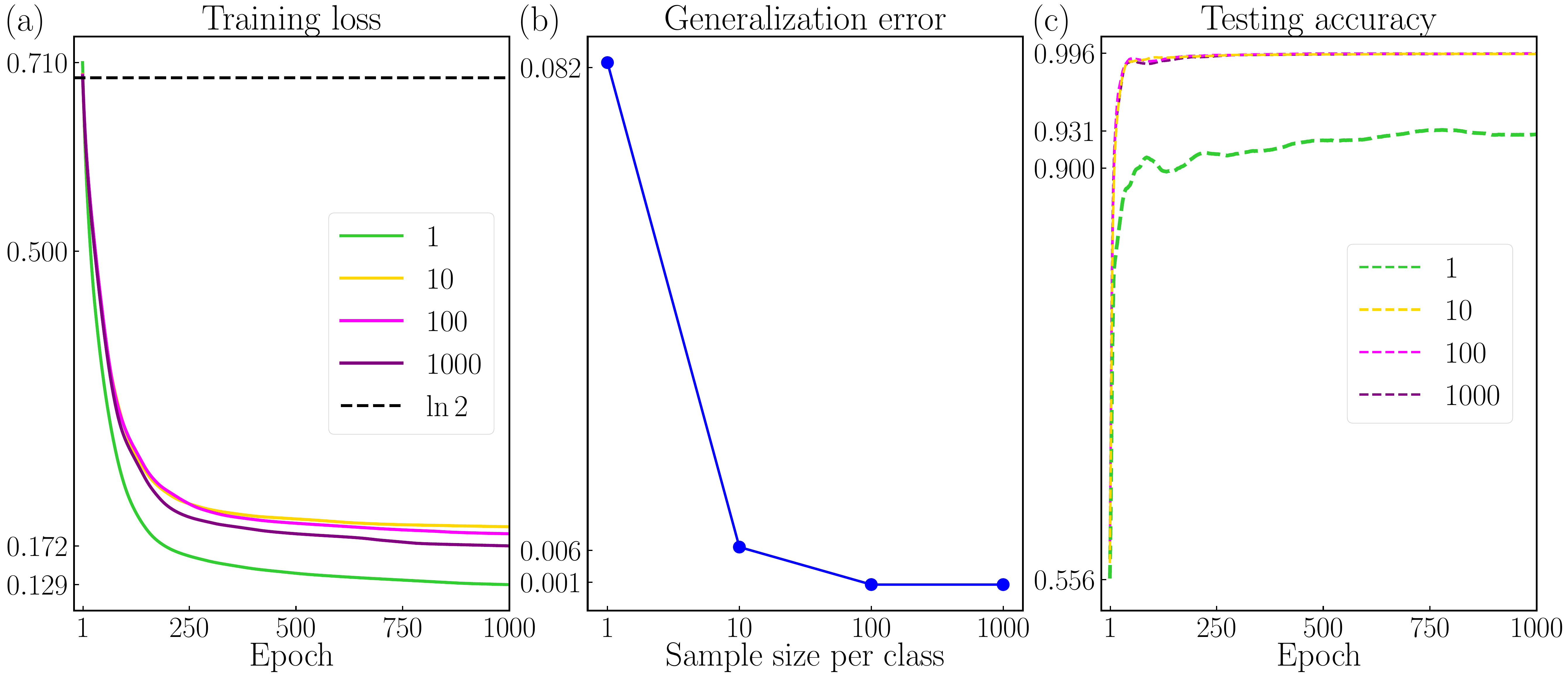}
	\caption{Quantum classification performance for MNIST dataset. (a) The average training loss over 10 runs for different sample sizes per class. (b) The average generalization error over 10 runs versus different sample sizes per class. (c) The average testing accuracy over 10 runs for different sample sizes per class. }
	\label{fig:mnist_loss}
\end{figure*}
As illustrated in Fig.~\ref{fig:real_datasets}, the trace distances between the averaged encoded states of the respective two classes in the EuroSAT, PathMNIST, and CIFAR-10 datasets consistently remain below 0.2 as the sample size increases. Notably, for the EuroSAT and PathMNIST datasets, the trace distances are even less than 0.023. In contrast, for the relatively simple and sparse MNIST dataset, the trace distance between the averaged encoded states of the two classes is close to 0.9. The comparisons in Fig.~\ref{fig:real_datasets} imply that good performance of quantum classifiers under amplitude encoding on the MNIST dataset may not be generalized to more complex datasets. 

To see this, we employ the quantum convolutional neural network (QCNN) as the PQC to train a quantum classifier to categorize the EuroSAT dataset. The QCNN has been shown to be free from  barren plateaus~\cite{pesah2021absence}, and has been widely used for quantum classification tasks. We illustrate the QCNN circuit in  Appendix~\ref{app:QCNN}. Similar to Sec.~\ref{sub:trainability}, we employ the Adam as the optimizer with a learning rate of 0.01. For class 1, we use the Pauli Z operator on the first qubit as the observable, while for class 2, we use the Pauli X operator on the first qubit as the observable. We ensure that the two classes have the same number of samples. We conduct simulations for 10 rounds, and in each round the parameters of the QCNN are independently initialized following the  standard normal distribution $\mathcal{N}(0,1)$. After training, we employ a test set consisting of 4000 new, unseen samples (2000 per class), and calculate the accuracy on this test set as the prediction accuracy.

We depict the performance of the trained quantum classifier in Fig.~\ref{fig:eurosat_loss}. From Fig.~\ref{fig:eurosat_loss}(a), as the sample size increases, the average convergence value of its training loss increases and approaches to $\ln 2$. Fig.~\ref{fig:eurosat_loss}(b) illustrates the average of the generalization error (Eq.~\eqref{eq:gen}) over 10 runs with 1000 iterations of each run (10000 results in total). We find that the generalization error decreases as the sample size increases. From Fig.~\ref{fig:eurosat_loss}(c), we find that for different sample sizes, the average prediction accuracy remains around 0.5 and never exceeds 0.56, indicating poor classification performance on the EuroSAT dataset under amplitude encoding.

Then for comparison, we also perform numerical simulations on the MNIST dataset under the same settings as those for the EuroSAT dataset, including the circuit, optimizer, learning rate, parameter initializations, and observables. We demonstrate the results  in Fig.~\ref{fig:mnist_loss}. Notably, the training losses on the MNIST dataset are significantly lower than ln2  across  different sample sizes, which is in stark contrast to  the EuroSAT dataset case. Most remarkably, a prediction accuracy above 0.9 can be achieved with just one sample per class, and when the number of samples per class exceeds 10, the testing accuracy remains consistently above 0.99. The success of classification under amplitude encoding on the MNIST dataset can be attributed to the sparsity of its features. The comparison between the MNIST and EuroSAT datasets indicates that findings validated on the MNIST dataset may not be readily generalized to more complex, real-world datasets.

\section{Conclusion}\label{IV}

We have investigated the concentration phenomenon induced by amplitude encoding and the resulting loss barrier phenomenon. Since encoding is a necessary and crucial step in leveraging QML to address classical problems, our findings indicate that the direct use of amplitude encoding may undermine the potential advantages of  QML. Therefore, more effort should be devoted to developing more efficient encoding strategies to fully unlock the potential of QML.

\section*{Acknowledgement}
This work was supported by the National Key Research and Development Program of China
(No. 2024YFA1013104).
\bibliography{ref}

\onecolumngrid

\pagebreak

\appendix

\setcounter{lemma}{0}
\renewcommand{\thelemma}{A.\arabic{lemma}}
\setcounter{proposition}{0}
\renewcommand{\theproposition}{A.\arabic{proposition}}
\setcounter{theorem}{0}
\renewcommand{\thetheorem}{A.\arabic{theorem}}
\setcounter{corollary}{0}
\renewcommand{\thecorollary}{A.\arabic{corollary}}

\section{Proof of Theorem~\ref{thm:loss_barrier}}
\label{app:proof-loss-barrier}

\begin{theorem}[Theorem~\ref{thm:loss_barrier} in the main text]
	\label{athm:loss_barrier}
	For a $K$-class classification, we employ the cross-entropy loss function $\mathcal{L}_S(\boldsymbol{\theta})$  defined in Eq.~\eqref{eq:celoss}. The quantum classifier is trained on a balanced training set $S = \{(\boldsymbol{x}^{(m)},y^{(m)})\}_{m=1}^{M}$, where  each class contains $M/K$ samples. Suppose the eigenvalues of each observable $H_k$ belong to $[-1,1]$, for $k=1,\ldots, K$. If the trace distance between the expectations of encoded states of any different classes is less than  $\epsilon$, then for any PQC $U(\boldsymbol{\theta})$ and optimization algorithm, we have
	\begin{equation}
		\label{eq:trace_between_class}
	\begin{aligned}
	\mathcal{L}_S(\boldsymbol{\theta}) \geqslant \ln\left[K - 4(K-1)\epsilon \right]
	\end{aligned}
	\end{equation}
	with probability at least $1 - 8e^{-M \epsilon^2 /8K}$.
\end{theorem}

\begin{proof}
	For the training set $S$, we denote by $\boldsymbol{x}^{(k,m)}$  the $m$-th feature of the $k$-th class, and $\boldsymbol{x}^{(j,n)}$ the $n$-th feature of the $j$-th class. For the observable $H_l$, the corresponding expectations of measurements read	
\begin{equation*}
\begin{aligned}
	h_{l}^{(k,m)} & \coloneqq  \operatorname{Tr}\left[H_l U(\boldsymbol{\theta}) \rho (\boldsymbol{x}^{(k,m)} ) U^{\dagger} \left(\boldsymbol{\theta} \right)\right], \\
	h_{l}^{(j,n)} & \coloneqq  \operatorname{Tr}\left[H_l U\left( \boldsymbol{\theta} \right) \rho( \boldsymbol{x}^{(j,n)})U^{\dagger}(\boldsymbol{\theta})   \right].
\end{aligned}
\end{equation*}
According to the assumption, the training set $S$ contains $M$ samples in total, with $M' = M/K$ samples per class. For the $k$-th class, we define the averaged expectation value over the subset of training samples $\{(\boldsymbol{x}^{(k,m)},k)\}_{m = 1}^{M'}$ as  $\overline{h}_l^{(k)} \coloneqq \frac{1}{M'} \sum_{m = 1}^{M'} h_l^{(k,m)}$. Similarly, we define the averaged encoded state $\overline{\rho}_{M'}^{(k)} \coloneqq \frac{1}{M'}\sum_{m=1}^{M'} \rho(\boldsymbol{x}^{(k,m)})$ for this subset. For the $j$-th class, we  define $\overline{h}_l^{(j)} \coloneqq \frac{1}{M'} \sum_{n = 1}^{M'} h_l^{(j,n)}$ and $\overline{\rho}_{M'}^{(j)} \coloneqq \frac{1}{M'} \sum_{n=1}^{M'} \rho(\boldsymbol{x}^{(j,n)})$, in a similar manner.

By applying H\"{o}lder's inequality, we have
$$
\begin{aligned}
\left|\overline{h}_l^{(k)} - \overline{h}_l^{(j)}\right| &= \left| \operatorname{Tr}\left[H_l U(\boldsymbol{\theta}) \left( \overline{\rho}_{M'}^{(k)}  - \overline{\rho}_{M'}^{(j)} \right) U^{\dagger}(\boldsymbol{\theta})  \right] \right| \\ 
& \leqslant \Big\|U^{\dagger}(\boldsymbol{\theta})H_l U(\boldsymbol{\theta})\Big\|_{\infty}  \left\|\overline{\rho}_{M'}^{(k)} - \overline{\rho}_{M'}^{(j)}\right\|_{1} \\
& \leqslant  \left\| \overline{\rho}_{M'}^{(k)} - \underset{\boldsymbol{x}^{(k)} \sim \mathcal{D}_{\mathcal{X}_k}} {\mathbb{E}}[\rho(\boldsymbol{x}^{(k)})]  + \underset{\boldsymbol{x}^{(k)} \sim \mathcal{D}_{\mathcal{X}_k}} {\mathbb{E}}[\rho(\boldsymbol{x}^{(k)})]  - \overline{\rho}_{M'}^{(j)} \right\|_{1}  \\
& \leqslant  \left\| \overline{\rho}_{M'}^{(k)} - \underset{\boldsymbol{x}^{(k)} \sim \mathcal{D}_{\mathcal{X}_k}} {\mathbb{E}}[\rho(\boldsymbol{x}^{(k)})]  \right\|_{1} + \left\| \underset{\boldsymbol{x}^{(k)} \sim \mathcal{D}_{\mathcal{X}_k}} {\mathbb{E}}[\rho(\boldsymbol{x}^{(k)})]  - \overline{\rho}_{M'}^{(j)} \right\|_{1}  \\
& \leqslant \left\| \overline{\rho}_{M'}^{(k)} - \underset{\boldsymbol{x}^{(k)} \sim \mathcal{D}_{\mathcal{X}_k}} {\mathbb{E}}[\rho(\boldsymbol{x}^{(k)})]  \right\|_{1} + \left\| \underset{\boldsymbol{x}^{(k)} \sim \mathcal{D}_{\mathcal{X}_k}} {\mathbb{E}}[\rho(\boldsymbol{x}^{(k)})]   - \underset{\boldsymbol{x}^{(j)} \sim \mathcal{D}_{\mathcal{X}_j}} {\mathbb{E}}[\rho(\boldsymbol{x}^{(j)})] + \underset{\boldsymbol{x}^{(j)} \sim \mathcal{D}_{\mathcal{X}_j}} {\mathbb{E}}[\rho(\boldsymbol{x}^{(j)})]- \overline{\rho}_{M'}^{(j)} \right\|_{1} \\
& \leqslant \left\| \overline{\rho}_{M'}^{(k)} - \underset{\boldsymbol{x}^{(k)} \sim \mathcal{D}_{\mathcal{X}_k}} {\mathbb{E}}[\rho(\boldsymbol{x}^{(k)})]  \right\|_{1}  + \left\|  \underset{\boldsymbol{x}^{(k)} \sim \mathcal{D}_{\mathcal{X}_k}} {\mathbb{E}}[\rho(\boldsymbol{x}^{(k)})]   - \underset{\boldsymbol{x}^{(j)} \sim \mathcal{D}_{\mathcal{X}_j}} {\mathbb{E}}[\rho(\boldsymbol{x}^{(j)})] \right\|_{1} + \left\| \underset{\boldsymbol{x}^{(j)} \sim \mathcal{D}_{\mathcal{X}_j}} {\mathbb{E}}[\rho(\boldsymbol{x}^{(j)})]- \overline{\rho}_{M'}^{(j)}  \right\|_{1}.
\end{aligned}
$$
According to Lemma~\ref{lem:hoffding}, it yields
$$
\begin{aligned}
	\left\| \overline{\rho}_{M'}^{(k)} - \underset{\boldsymbol{x}^{(k)} \sim \mathcal{D}_{\mathcal{X}_k}} {\mathbb{E}}[\rho(\boldsymbol{x}^{(k)})]  \right\|_{1} + \left\| \underset{\boldsymbol{x}^{(j)} \sim \mathcal{D}_{\mathcal{X}_j}} {\mathbb{E}}[\rho(\boldsymbol{x}^{(j)})]- \overline{\rho}_{M'}^{(j)}  \right\|_{1} \leqslant 2 \epsilon,
\end{aligned}
$$
with probability at least $1-8 e^{-M' \epsilon^2 /8}$. Combining this with our assumption that
$$
\begin{aligned}
	\left\|  \underset{\boldsymbol{x}^{(k)} \sim \mathcal{D}_{\mathcal{X}_k}} {\mathbb{E}}[\rho(\boldsymbol{x}^{(k)})]   - \underset{\boldsymbol{x}^{(j)} \sim \mathcal{D}_{\mathcal{X}_j}} {\mathbb{E}}[\rho(\boldsymbol{x}^{(j)})] \right\|_{1}  \leqslant 2\epsilon,
\end{aligned}
$$
we have
\begin{equation}
	\label{eq:h-h}
	\begin{aligned}
		\left|\overline{h}_l^{(k)} - \overline{h}_l^{(j)}\right| &= \left| \operatorname{Tr}\left[H_l U(\boldsymbol{\theta}) \left( \overline{\rho}_{M'}^{(k)}  - \overline{\rho}_{M'}^{(j)} \right) U^{\dagger}(\boldsymbol{\theta})  \right] \right| \leqslant 4\epsilon,
	\end{aligned}
\end{equation}
with probability at least $1-8e^{-M \epsilon^2 / 8K}$.

We now derive the lower bound for the cross-entropy function. According to Eq.~\eqref{eq:celoss}, the cross-entropy reads
\begin{align}
	 \mathcal{L}_S(\boldsymbol{\theta}) &=  \frac{1}{M} \sum_{m=1}^{M} \ell (\boldsymbol{\theta};\boldsymbol{x}^{(m)},y^{(m)})  \nonumber \\
	&= - \frac{1}{M} \sum_{m=1}^{M} \sum_{k=1}^{K} \boldsymbol{y}_{k}^{(m)} \ln \left( \frac{\mathrm{e}^{h_{k}(\boldsymbol{x}^{(m)},\boldsymbol{\theta})}}{\sum_{j=1}^{K} \mathrm{e}^{h_{j}(\boldsymbol{x}^{(m)},\boldsymbol{\theta})}}  \right)                    \nonumber \\
	&= \frac{1}{M} \sum_{m=1}^{M} \sum_{k=1}^{K} \boldsymbol{y}_{k}^{(m)} \ln \left( \frac{  \sum_{j=1}^{K} \mathrm{e}^{h_{j}(\boldsymbol{x}^{(m)},\boldsymbol{\theta})} }{\mathrm{e}^{h_{k}(\boldsymbol{x}^{(m)},\boldsymbol{\theta})}} \right)   \nonumber  \\
	&= \frac{1}{M} \sum_{m=1}^{M} \sum_{k=1}^{K} \boldsymbol{y}_k^{(m)} \ln\left[ \sum_{j=1}^{K} \exp\left\{ h_j (\boldsymbol{x}^{(m)},\boldsymbol{\theta}) - h_k(\boldsymbol{x}^{(m)},\boldsymbol{\theta}) \right\} \right] \nonumber \\
	&= \frac{1}{K} \sum_{k=1}^{K} \frac{1}{M'} \sum_{m =1}^{M'}  \ln\left[ \sum_{j=1}^{K} \exp \left\{ h_j(\boldsymbol{x}^{(k,m)},\boldsymbol{\theta}) - h_k(\boldsymbol{x}^{(k,m)},\boldsymbol{\theta}) \right\} \right], \label{eq:Mk} 
\end{align}
where Eq.~\eqref{eq:Mk} arises from the fact that only the feature $\boldsymbol{x}^{(m)}$ from the $k$-th class has $\boldsymbol{y}_k^{(m)} \neq 0$, and in this case, $\boldsymbol{y}_k^{(m)} = 1$.

We then apply Jensen's inequality to the convex function $\ln[\sum_{j=1}^{K} e^{x_j}]$ on $\mathbb{R}^{K}$~\cite{boyd2004convex}. It yields
\begin{align}
\mathcal{L}_S(\boldsymbol{\theta}) &= \frac{1}{K} \sum_{k=1}^{K} \frac{1}{M'} \sum_{m=1}^{M'} \ln\left[ \sum_{j=1}^{K} \exp \left\{ h_j^{(k,m)} - h_k^{(k,m)} \right\} \right] \nonumber \\ &
\geqslant \frac{1}{K} \sum_{k=1}^{K} \ln \left[  \sum_{j=1}^{K} \exp \left\{ \frac{1}{M'} \sum_{m=1}^{M'} \left( h_j^{(k,m)} - h_k^{(k,m)} \right)  \right\} \right] \label{eq:Jensen_1}\\
&= \frac{1}{K} \sum_{k=1}^{K} \ln \left[ \sum_{j=1}^{K} \exp \left\{ \overline{h} _j^{(k)} - \overline{h} _k^{(k)} \right\} \right] \nonumber \\
& \geqslant \ln \left[ \sum_{j=1}^{K} \exp \left\{ \frac{1}{K} \sum_{k=1}^{K} \left( \overline{h} _j^{(k)} - \overline{h} _k^{(k)} \right) \right\}  \right] \label{eq:Jensen_2}  \\
& \geqslant \ln \left[ \sum_{j=1}^{K} \left\{ \frac{1}{K} \sum_{k=1}^{K} \left( \overline{h}_j^{(k)} - \overline{h}_k^{(k)} \right) +1  \right\}  \right] \label{eq:ex} \\
&= \ln \left[ \frac{1}{K} \sum_{j=1}^{K} \sum_{k=1}^{K} \left( \overline{h}_j^{(k)} - \overline{h} _k^{(k)} \right) + K   \right]  \label{eq:next},
\end{align}
where inequalities~\eqref{eq:Jensen_1} and~\eqref{eq:Jensen_2} are derived from Jensen's inequality, and inequality~\eqref{eq:ex} follows from the  inequality $e^{x} \geqslant x+1$.

Since Eq.~\eqref{eq:next} involves a double summation over indices $j$ and $k$, we can change the subscripts of $h_k^{(k)}$ to $h_{j}^{(j)}$. This yields
$$
\begin{aligned}
\mathcal{L}_S(\boldsymbol{\theta}) & \geqslant \ln \left[ \frac{1}{K} \sum_{j=1}^{K} \sum_{k=1}^{K} \left( \overline{h}_j^{(k)} - \overline{h} _k^{(k)} \right) + K   \right] \\
&= \ln\left[ \frac{1}{K} \sum_{k=1}^{K} \sum_{j=1}^{K} \left( \overline{h}_j^{(k)} - \overline{h}_j^{(j)}  \right) + K \right].  \\
\end{aligned}
$$
Finally, combining this with Eq.~\eqref{eq:h-h}, we have$$
\begin{aligned}
\mathcal{L}_S(\boldsymbol{\theta}) &\geqslant \ln \left[ K - 4(K-1) \epsilon  \right] \\
\end{aligned}
$$
with probability at least $1 - 8 e^{- M \epsilon^2 /8K}$.
\end{proof}

\section{Proofs of concentrations}
\label{app:concentration-proof}

\setcounter{lemma}{0}
\renewcommand{\thelemma}{B.\arabic{lemma}}
\setcounter{proposition}{0}
\renewcommand{\theproposition}{B.\arabic{proposition}}
\setcounter{theorem}{0}
\renewcommand{\thetheorem}{B.\arabic{theorem}}
\setcounter{corollary}{0}
\renewcommand{\thecorollary}{B.\arabic{corollary}}

\begin{lemma}[Lemma~\ref{lem:hoffding} in the main text]
	\label{alem:hoffding}
	Given an arbitrary $\epsilon \in (0,1)$, we have
	$$
	\begin{aligned}
	 T \left( \overline{\rho}_M,\mathbb{E}[\rho] \right)  \leq \epsilon 
	\end{aligned}
	$$
	with probability at least $1 - 4 e^{-M \epsilon^{2}/2}$. 
\end{lemma}

\begin{proof}
	For any linear operator $X$ defined on a linear space $\mathcal{X}$, its Schatten-1 norm reads
	\begin{equation*}
	\begin{aligned}
		\left\| X \right\|_{1} = \max_{U \in \mathrm{U}(\mathcal{X})}  |\langle U,X \rangle |,
	\end{aligned}
	\end{equation*}
	where $\mathrm{U}(\mathcal{X})$ denotes the set of unitary operators on space $\mathcal{X}$. Consequently, we can express $$\left\| \bar{\rho}_{M} - \mathbb{E}[\rho] \right\|_{1} = \max_{U \in \mathrm{U}(\mathcal{X})}|\langle U,\bar{\rho}_{M} - \mathbb{E}[\rho] \rangle |,$$ and 
	\begin{equation*}
	\begin{aligned}
		|\langle U,\bar{\rho}_{M} - \mathbb{E}[\rho] \rangle | = |\mathrm{Tr}\left[U^{\dagger} (\bar{\rho}_{M} - \mathbb{E}[\rho])\right] | = | \mathrm{Tr}\left[U^{\dagger} \bar{\rho}_{M}\right]  - \mathrm{Tr}\left[U^{\dagger}\mathbb{E}[\rho]\right] |.
	\end{aligned}
	\end{equation*}
	The complex-valued term $\mathrm{Tr}\left[U^{\dagger} \bar{\rho}_{M}\right]$ can be decomposed into its real and imaginary components as
	\begin{equation*}
	\begin{aligned}
		\mathrm{Tr}\left[U^{\dagger} \bar{\rho}_{M}\right] = \mathrm{Re}\{ \mathrm{Tr}\left[U^{\dagger} \bar{\rho}_{M}\right] \} + \mathrm{i}  \, \mathrm{Im}\{ \mathrm{Tr}\left[U^{\dagger} \bar{\rho}_{M}\right] \}.
	\end{aligned}
	\end{equation*}
	For the real component, we have
	\begin{equation*}
	\begin{aligned}
		\mathrm{Re}\{ \mathrm{Tr}\left[U^{\dagger} \bar{\rho}_{M} \right] \} = \frac{1}{M} \sum_{m=1}^{M} \mathrm{Re}\left\{  \mathrm{Tr}\left[U^{\dagger} \rho(\boldsymbol{x}^{(m)})\right] \right\}.
	\end{aligned}
	\end{equation*}
	Given that \( \left| \mathrm{Tr} \left[ U^{\dagger} \bar{\rho}_M \right] \right| \leq \| U^{\dagger} \|_{\infty} \| \rho \|_1 = 1 \), and considering that for any complex number \( z = a + \mathrm{i}b \), \( |z| = \sqrt{a^2 + b^2} \geq |a| \), we can see that \( \frac{1}{M} \mathrm{Re} \left\{ \mathrm{Tr} \left[ U^{\dagger} \rho(\boldsymbol{x}^{(m)}) \right] \right\} \) is an independent random variable bounded within \( \left[ -\frac{1}{M}, \frac{1}{M} \right] \). Applying Hoeffding's inequality~\cite{mohri2018foundations} yields
	\begin{equation*}
	\begin{aligned}
		\left| \frac{1}{M} \sum_{m=1}^{M} \mathrm{Re} \left\{ \mathrm{Tr} \left[ U^{\dagger} \rho(\boldsymbol{x}^{(m)}) \right] \right\} - \underset{\boldsymbol{x} \sim \mathcal{D}_{\mathcal{X}}}{\mathbb{E}} \left( \mathrm{Re} \left\{ \mathrm{Tr} \left[ U^{\dagger} \rho(\boldsymbol{x}) \right] \right\} \right) \right| \leq \epsilon,
	\end{aligned}
	\end{equation*}
	with probability at least \( 1 - 2 e^{-M \epsilon^2 / 2} \). 
	
By the definitions of \( \bar{\rho}_M \) and \( \mathbb{E}[\rho] \), we obtain $\left| \mathrm{Re} \left\{ \mathrm{Tr} \left[ U^{\dagger} \bar{\rho}_M \right] - \mathrm{Tr} \left[ U^{\dagger} \mathbb{E}[\rho] \right] \right\} \right| \leq \epsilon$
	with probability at least \( 1 - 2 e^{-M \epsilon^2 / 2} \). Following a similar argument, we have $\left| \mathrm{Im} \left\{ \mathrm{Tr} \left[ U^{\dagger} \bar{\rho}_M \right] - \mathrm{Tr} \left[ U^{\dagger} \mathbb{E}[\rho] \right] \right\} \right| \leq \epsilon$ with the same probability  at least \( 1 - 2 e^{-M \epsilon^2 / 2} \).
	
	Therefore, by the triangle inequality, we have
	$$
	\begin{aligned}
		\left| \mathrm{Tr} \left[ U^{\dagger} \bar{\rho}_M \right] - \mathrm{Tr} \left[ U^{\dagger} \mathbb{E}[\rho] \right] \right| &= \left| \mathrm{Re} \left\{ \mathrm{Tr} \left[ U^{\dagger} \bar{\rho}_M \right] - \mathrm{Tr} \left[ U^{\dagger} \mathbb{E}[\rho] \right] \right\} + \mathrm{i} \, \mathrm{Im} \left\{ \mathrm{Tr} \left[ U^{\dagger} \bar{\rho}_M \right] - \mathrm{Tr} \left[ U^{\dagger} \mathbb{E}[\rho] \right] \right\} \right| \\
		&\leq \left| \mathrm{Re} \left\{ \mathrm{Tr} \left[ U^{\dagger} \bar{\rho}_M \right] - \mathrm{Tr} \left[ U^{\dagger} \mathbb{E}[\rho] \right] \right\} \right| + \left| \mathrm{Im} \left\{ \mathrm{Tr} \left[ U^{\dagger} \bar{\rho}_M \right] - \mathrm{Tr} \left[ U^{\dagger} \mathbb{E}[\rho] \right] \right\} \right| \\
		&\leq 2 \epsilon,
	\end{aligned}
	$$
	with probability at least \( 1 - 4 e^{-M \epsilon^2 / 2} \). Finally, we can conclude that
	$$
	\begin{aligned}
		T(\bar{\rho}_M, \mathbb{E}[\rho]) = \frac{1}{2} \left\| \bar{\rho}_M - \mathbb{E}[\rho] \right\|_1 = \frac{1}{2} \max_{U \in \mathrm{U}(\mathcal{X})} \left| \langle U, \bar{\rho}_M - \mathbb{E}[\rho] \rangle \right| \leq \epsilon,
	\end{aligned}
	$$
	with probability at least \( 1 - 4 e^{-M \epsilon^2 / 2} \).

\end{proof}

\begin{proposition}[Proposition~\ref{pro:min-max-con} in the main text]
	\label{pro: min-max-con}
Assume that all elements in the feature $\boldsymbol{x} \in \mathbb{R}^{2^{n}}$ have the same sign,  and the elements satisfy $|x_i|\in[m, M]$. If $\left|\frac{m}{M} -1 \right| < \epsilon$, then after amplitude encoding, we have
	\begin{equation*}
	\begin{aligned}
		T \left( \rho(\boldsymbol{x}) ,\frac{1}{2^{n}} \mathbf{1} \mathbf{1}^{\top}  \right) \leqslant \sqrt{2 \epsilon},
	\end{aligned}
	\end{equation*}
where the state $\frac{1}{2^{n}} \boldsymbol{1} \boldsymbol{1}^{\top} = \ket{+}_{n} \bra{+}_{n}$, with $\ket{+}_n = H^{\otimes n} \ket{0}_{n} $ being the superposition of all computational  basis states.
\end{proposition}		

\begin{proof}
	For the feature vector  $\boldsymbol{x} = [x_0,x_1,\cdots,x_{2^{n}-1}]^{\top}$, denote that the encoded state vector as $\ket{\psi(\boldsymbol{x})}$, with the density matrix as $\rho(\boldsymbol{x})$. Then, the probability amplitude of the encoded state correspondingto the basis vector $\ket{i}$ reads
	$$
	\begin{aligned}
		\alpha_i = \sqrt{\frac{x_i^2}{x_0^2 + x_1^2 + \cdots x_{2^{n}-1}^2}}.
	\end{aligned}
	$$
The 	corresponding probability of projection into the basis vector $\ket{i}$ is
	$$
	\begin{aligned}
		p_i = |\alpha_i|^2 = \frac{x_i^2}{x_0^2 + x_1^2 + \cdots x_{2^{n}-1}^2} = \frac{1}{1 + (x_0 / x_i)^2  + \cdots + (x_{2^{n}-1} / x_{i})^2 }.
	\end{aligned}
	$$
If  $|x_i|\in[m, M]$, then
		$$
	\begin{aligned}
		p_i = \frac{1}{1 + (x_0 / x_i)^2  + \cdots + (x_{2^{n}-1} / x_{i})^2 } \geqslant \frac{1}{1 + (2^{n}-1) (M / m)^2 }.
	\end{aligned}
	$$
	Let $x = \frac{m}{M}$, where $x \in (0,1]$, and define $f_{\min}(x) = \frac{x^2}{x^2 + (2^{n}-1)}$. Thus, we have
	$$
	\begin{aligned}
		\left| 2^nf_{\min}(x) - 1 \right| &= \left| \frac{2^{n} x^2 - x^2 - (2^{n}-1)}{x^2 + (2^{n}-1)} \right|  \\
        &= \left| \frac{(2^{n}-1)(x^2 - 1)}{x^2 + (2^{n}-1)} \right|  \\
        &\leqslant |x^2 - 1|\\
        & \leqslant 2 \epsilon.
	\end{aligned}
	$$
	Thus,  $f_{\min}(\frac{m}{M}) = \frac{1}{1 + (2^{n}-1) (M / m)^2 } \geqslant \frac{1}{2^{n}} \cdot (1 - 2 \epsilon).$

The absolute value of the inner product between the encoded state vector $\ket{\psi(\boldsymbol{x})}$ and the $n$-qubit superposition state $\ket{+}_n$ satisfies
	$$
	\begin{aligned}
	    \Big|\langle \psi(\boldsymbol{x})|+ \rangle_n \Big| &= \left| \sum_{i=0}^{2^{n}-1}  \frac{1}{\sqrt{2^{n}}} \frac{x_i}{\sqrt{x_0 ^2 + \cdots + x_{2^{n}-1}}} \right| \\
    &=  \frac{1}{\sqrt{2^{n}}} \sum_{i=0}^{2^{n}-1} \left| \frac{x_i}{\sqrt{x_0^2  + \cdots + x_{2^{n}-1}^2}}  \right| \\
    &= \frac{1}{\sqrt{2^{n}}} \sum_{i=0}^{2^{n}-1} \sqrt{\frac{1}{ 1 + (x_0 / x_i) ^2 + \cdots + (x_{2^{n}-1} / x_i)^2}}  \\
    &\geqslant \frac{1}{\sqrt{2^{n}}} \sum_{i=0}^{2^{n}-1} \sqrt{\frac{1}{1 + (2^{n} - 1) (M / m)^2}} \\
    & \geqslant \sum_{i=0}^{2^{n}-1} \frac{1}{\sqrt{2^{n}}} \cdot \frac{1}{\sqrt{2^{n}}} \cdot \sqrt{1 - 2 \epsilon} \\
    & \geqslant \sqrt{1 - 2 \epsilon}.
	\end{aligned}
	$$
	Since both the encoded state $\ket{\psi(\boldsymbol{x})}$ and the superposition state $\ket{+}_n$ are pure, we have
	$$
	\begin{aligned}
	    T \left( \rho(\boldsymbol{x}) ,\frac{1}{2^{n}} \mathbf{1} \mathbf{1}^{\top}  \right) = \frac{1}{2}\left\|\rho(\boldsymbol{x}) - \frac{1}{2^{n}} \mathbf{1} \mathbf{1}^{\top}  \right\|_{1}  &=  \sqrt{1 - |\langle \psi(\boldsymbol{x}) | + \rangle_n |^2} \leqslant \sqrt{2 \epsilon}.
	\end{aligned}
	$$
\end{proof}

\begin{proposition}[Proposition~\ref{pro:max_mixed} in the main text]
	Denote by $\mathcal{D}_{\mathcal{X}}$ the distribution of the feature $\boldsymbol{x} \in \mathbb{R}^{2^{n}}$. Assume that the elements of $\boldsymbol{x}$ are i.i.d with an expected value of $0$. In addition, the distribution is symmetric, i.e., $p(x_j) = p(-x_j)$, for all $j \in [0:2^{n}-1]$. Then after amplitude encoding, the expectation of  encoded state is the maximally mixed state, i.e.,
	\begin{equation*}
	\begin{aligned}
	\underset{\boldsymbol{x} \sim  \mathcal{D}_{\mathcal{X}}}{\mathbb{E}}\left[\rho(\boldsymbol{x})\right] = \frac{\boldsymbol{I}}{2^{n}}.
	\end{aligned}
	\end{equation*}
\end{proposition}

\begin{proof}
    For the feature vector $\boldsymbol{x} = [x_0,x_1,\cdots,x_{2^{n}-1}]^{\top}$, according to the definition of amplitude encoding in Eq.~\eqref{eq:amp}, the expectation of encoded states reads
    \begin{equation*}
    \begin{aligned}
        \underset{\boldsymbol{x} \sim \mathcal{D}_{\mathcal{X}}}{\mathbb{E}} \left[\rho (\boldsymbol{x}) \right] &= \underset{\boldsymbol{x} \sim \mathcal{D}_{\mathcal{X}}}{\mathbb{E}}\left(\sum_{i=0}^{2^{n}-1}  \sum_{i=0}^{2^{n}-1}  \frac{x_{i } x_{j}}{C_{\boldsymbol{x}}^2} \ket{i} \bra{j} \right)  \\
		&= \sum_{i=0}^{2^{n}-1}  \sum_{j=0}^{2^{n}-1}\underset{\boldsymbol{x} \sim \mathcal{D}_{\mathcal{X}}}{\mathbb{E}}\left( \frac{x_{i} x_{j}}{x_0^2 +  + \cdots + x_{2^{n}-1}^2}  \right) \ket{i} \bra{j}.  \\
    \end{aligned}
    \end{equation*}

    (1) Consider the diagonal elements:
    \begin{equation*}
        \begin{aligned}
            \underset{\boldsymbol{x} \sim \mathcal{D}_{\mathcal{X}}}{\mathbb{E}}\left( \frac{x_i^2}{x_0^2 +  + \cdots + x_{2^{n}-1}^2}  \right)  &= \underset{\boldsymbol{x} \sim \mathcal{D}_{\mathcal{X}}}{\mathbb{E}}\left( 1  \right) - \underset{\boldsymbol{x} \sim \mathcal{D}_{\mathcal{X}}}{\mathbb{E}}\left(  \frac{x_0^2 + \cdots + x_{i-1}^2 + x_{i+1}^2 + \cdots + x_{2^{n}-1}^2}{x_0^2  +\cdots + x_{2^{n}-1}^2} \right) \\
            &= 1 - \sum_{j=0,j\neq i}^{2^{n}-1}  \underset{\boldsymbol{x} \sim \mathcal{D}_{\mathcal{X}}}{\mathbb{E}}\left(\frac{x_j^2}{x_0^2  + \cdots + x_{2^{n}-1}^2}  \right).  \\
        \end{aligned}
        \end{equation*}
Note that the elements $x_j$s' are independent and identically distributed. We assume that $\mathbb{E}_{\boldsymbol{x} \sim \mathcal{D}_{\mathcal{X}}}\left(\frac{x_j^2}{x_0^2  + \cdots + x_{2^{n}-1}^2} \right) = a$. Then we have $a = 1 - (2^{n} -1) a$, and $a = \frac{1}{2^{n}}$, accordingly. Thus, for all diagonal elements, we have $\mathbb{E}_{\boldsymbol{x} \sim \mathcal{D}_{\mathcal{X}}}\left( \frac{x_{i}^2}{C_{\boldsymbol{x}}^2}  \right) = \frac{1}{2^{n}}$.

    (2) Consider the non-diagonal elements: 
    \begin{equation*}
    \begin{aligned}
    	\underset{\boldsymbol{x} \sim \mathcal{D}_{\mathcal{X}}}{\mathbb{E}}\left( \frac{x_{i} x_{j}}{C_{\boldsymbol{x}}^2}  \right) &= \underset{\boldsymbol{x} \sim \mathcal{D}_{\mathcal{X}}}{\mathbb{E}}\left( \frac{x_{i} x_{j}}{x_0^2 + \cdots + x_{2^{n}-1}^2}   \right). \\
    \end{aligned}
    \end{equation*}
    Since  $\frac{(-x_{i}) x_{j}}{x_0^2 + \cdots + x_{2^{n}-1}^2}  = -  \frac{x_{i} x_{j}}{x_0^2  + \cdots + x_{2^{n}-1}^2}$, and the distribution of $x_j$ is symmetric, we have
    \begin{equation*}
    \begin{aligned}
        \underset{\boldsymbol{x} \sim \mathcal{D}_{\mathcal{X}}}{\mathbb{E}}\left(\frac{x_{i} x_{j}}{C_{\boldsymbol{x}}^2}  \right) &= \underset{\boldsymbol{x} \sim \mathcal{D}_{\mathcal{X}}}{\mathbb{E}}\left(\frac{x_{i} x_{j}}{x_0^2  + \cdots + x_{2^{n}-1}^2}  \right) = 0.
    \end{aligned}
    \end{equation*}

    Therefore,    
	  \begin{equation*}
		\begin{aligned}
		\underset{\boldsymbol{x} \sim \mathcal{D}_{\mathcal{X}}}{\mathbb{E}}\left( \frac{x_{i} x_{j}}{C_{\boldsymbol{x}}^2}  \right)  = \begin{cases}
			\frac{1}{2^{n}} & i = j; \\
			0 &  i \neq j.
		\end{cases}
		\end{aligned}
	  \end{equation*}
 This implies that the expectation of encoded states is the maximally mixed state, i.e., $\underset{\boldsymbol{x} \sim \mathcal{D}_{\mathcal{X}}}{\mathbb{E}}\left[ \rho(\boldsymbol{x}) \right] = \frac{I}{2^{n}}$.
\end{proof}

\begin{proposition}[Proposition~\ref{pro:binary_situation} in the main text]
	For a binary classification dataset, denote by $\mathcal{D}_{\mathcal{X}_i}$ the distribution of feature $\boldsymbol{x}\in \mathbb{R}^{2^{n}}$ in class $i$, for $i=1,2$.  If for any $j \in [0:2^{n}-1]$, the probability density functions $p_{i,j}(x)$ of the $j$-th element in class $i$ satisfy  $p_{1,j}(x) = -p_{2,j}(x)$, then after amplitude encoding the expected states of the two classes are identical, i.e.,
	\begin{equation*}
		\begin{aligned}
		\underset{\boldsymbol{x} \sim \mathcal{D}_{\mathcal{X}_1}}{\mathbb{E}}\left[\rho(\boldsymbol{x})\right] = \underset{\boldsymbol{x} \sim \mathcal{D}_{\mathcal{X}_2}}{\mathbb{E}}\left[\rho(\boldsymbol{x})\right].
		\end{aligned}
	\end{equation*}
\end{proposition}

\begin{proof}
    For the feature vector $\boldsymbol{x} = [x_0,x_1,\cdots,x_{2^{n}-1}]^{\top}$, let $[\rho(\boldsymbol{x})]_{ij}$ denote the $(i,j)$-th element of $\rho(\boldsymbol{x})$. Then, according to the definition of amplitude encoding in Eq.~\eqref{eq:amp}, we have 
	$$
	\begin{aligned}
	&\underset{\boldsymbol{x} \sim \mathcal{D}_{\mathcal{X}_1}}{\mathbb{E}}\left\{  [\rho(\boldsymbol{x})]_{ij} \right\} = \underset{\boldsymbol{x} \sim \mathcal{D}_{\mathcal{X}_1}}{\mathbb{E}}\left( \frac{x_{i} x_{j}}{C_{\boldsymbol{x}}^2}  \right) \\
	  =& \int_{}^{ }  \frac{x_i x_j}{x_{0}^2 + \cdots + x_{2^{n}-1}^2} p_{1,0}(x_0)  \cdots  p_{1,2^{n}-1}(x_{2^{n}-1}) ~\mathrm{d}x_0 \cdots  ~\mathrm{d}x_{2^{n}-1}  \\
	  =& (-1)^{2^{n}}\int_{}^{ }   \frac{x_i x_j}{x_{0}^2  + \cdots + x_{2^{n}-1}^2} p_{2,0}(x_0)\cdots p_{2,2^{n}-1}(x_{2^{n}-1}) ~\mathrm{d}x_0 \cdots   ~\mathrm{d}x_{2^{n}-1}  \\
	  =&\int_{}^{ }   \frac{x_i x_j}{x_{0}^2  + \cdots + x_{2^{n}-1}^2} p_{2,0}(x_0) \cdots p_{2,2^{n}-1}(x_{2^{n}-1}) ~\mathrm{d}x_0 \cdots  ~\mathrm{d}x_{2^{n}-1}  \\
	  =&\underset{\boldsymbol{x} \sim \mathcal{D}_{\mathcal{X}_2}}{\mathbb{E}}\left\{  [\rho(\boldsymbol{x})]_{ij} \right\}. \\ 
	\end{aligned}
	$$
	
	Therefore, 
$
	\begin{aligned}
	 \underset{\boldsymbol{x} \sim \mathcal{D}_{\mathcal{X}_1}}{\mathbb{E}}\left[\rho(\boldsymbol{x})\right] =  \underset{\boldsymbol{x} \sim \mathcal{D}_{\mathcal{X}_2} }{\mathbb{E}}\left[\rho(\boldsymbol{x})\right].
	\end{aligned}
	$

	\end{proof}
	
\section{More numerical validations}
\label{app:trainability-validation}

\setcounter{figure}{0}
\renewcommand{\thefigure}{C.\arabic{figure}}
We employ the same numerical settings as those in Subsec.~\ref{sub:trainability} and consider three different layer sizes: $L=1, 10$, and $30$, for the PQC illustrated in Fig.~\ref{fig:training_loss}.(a). For each layer size, we train a quantum classifier, and  depict the training loss and accuracy for the three datasets (Fig.~\ref{fig:example_1}(a), Fig.~\ref{fig:example_2}(a), and Fig.~\ref{fig:example_3}(a))  in Fig.~\ref{fig:loss_1_layer}, Fig.~\ref{fig:loss_10_layer}, and Fig.~\ref{fig:loss_30_layer}, respectively. It is clear that the training loss quickly converges to $\ln 2$, which corresponds to the random classification case, and the training accuracy remains around 0.5.

\begin{figure}[htpb]
	\centering
	\includegraphics[width=0.86\textwidth]{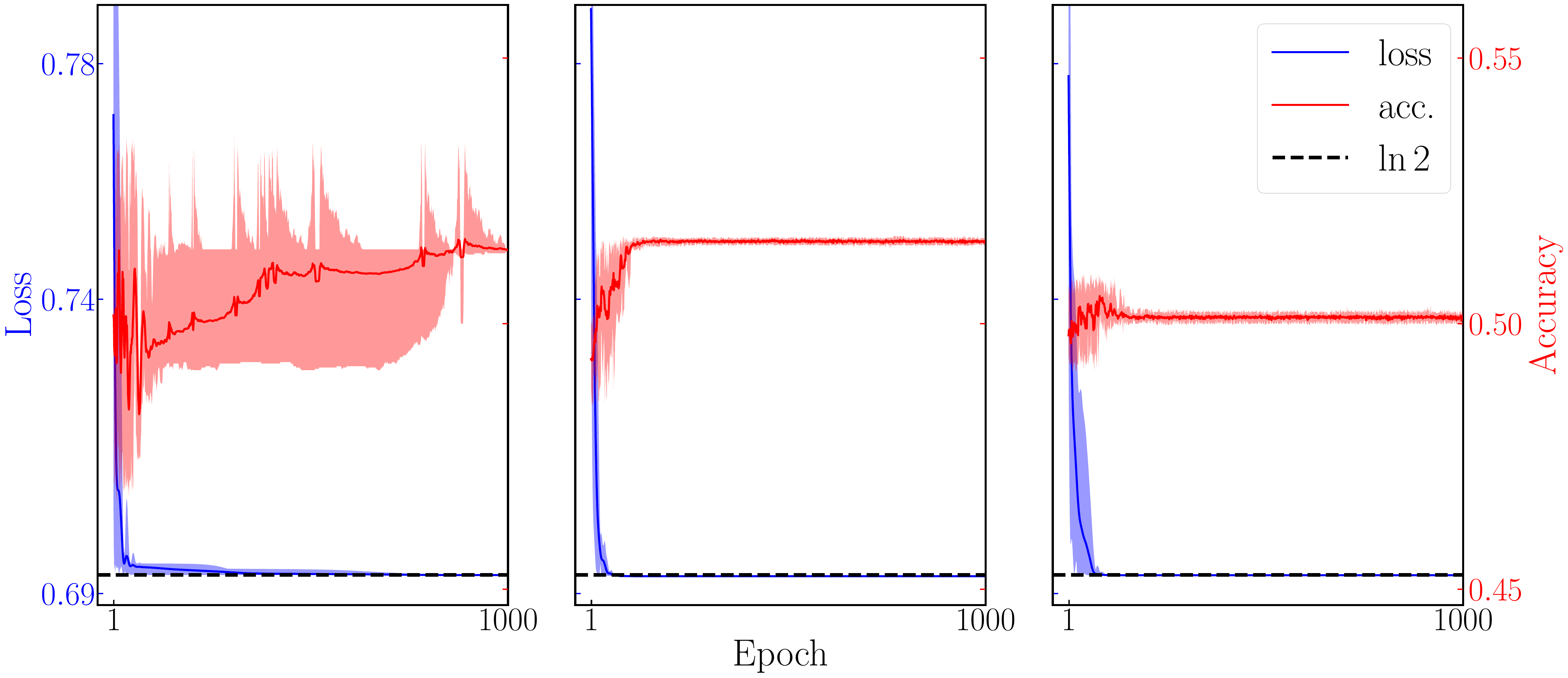}
	\caption{Performance of the quantum classifier trained by a 1-layer circuit. From left to right, the figures show the training loss and accuracy for the datasets in Fig.~\ref{fig:example_1}(a), Fig.~\ref{fig:example_2}(a), and Fig.~\ref{fig:example_3}(a), respectively. The solid blue and red line, respectively, represent the mean of the training loss and accuracy over 10 runs, and the shaded areas represent their respective ranges  across the 10 runs.}
	\label{fig:loss_1_layer}
\end{figure}

\begin{figure}[htpb]
	\centering
	\includegraphics[width=0.86\textwidth]{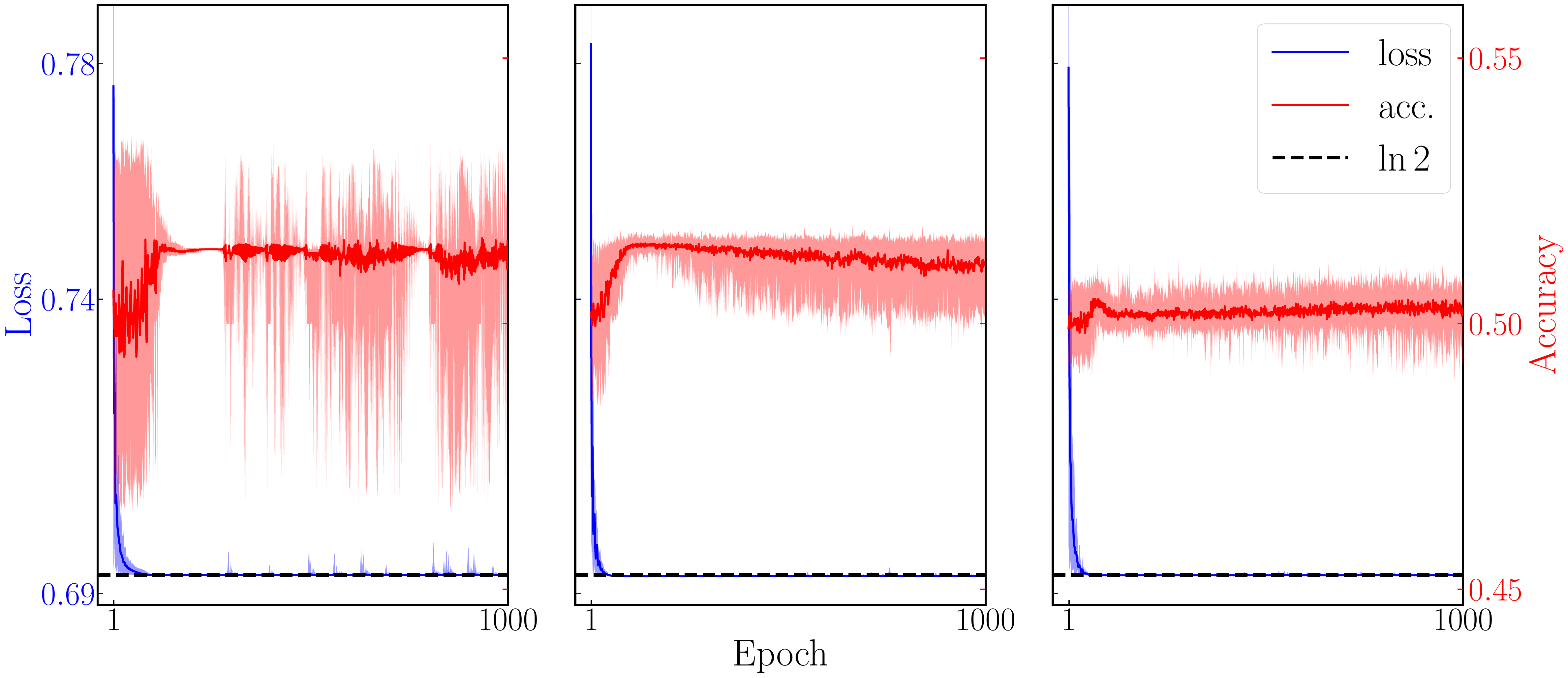}
	\caption{Performance of the quantum classifier trained by a 10-layer circuit. From left to right, the figures show the training loss and accuracy for the datasets in Fig.~\ref{fig:example_1}(a), Fig.~\ref{fig:example_2}(a), and Fig.~\ref{fig:example_3}(a), respectively. The solid blue and red line, respectively, represent the mean of the training loss and accuracy over 10 runs, and the shaded areas represent their respective ranges  across the 10 runs.}
	\label{fig:loss_10_layer}
\end{figure}

\begin{figure}[htpb]
	\centering
	\includegraphics[width=0.86\textwidth]{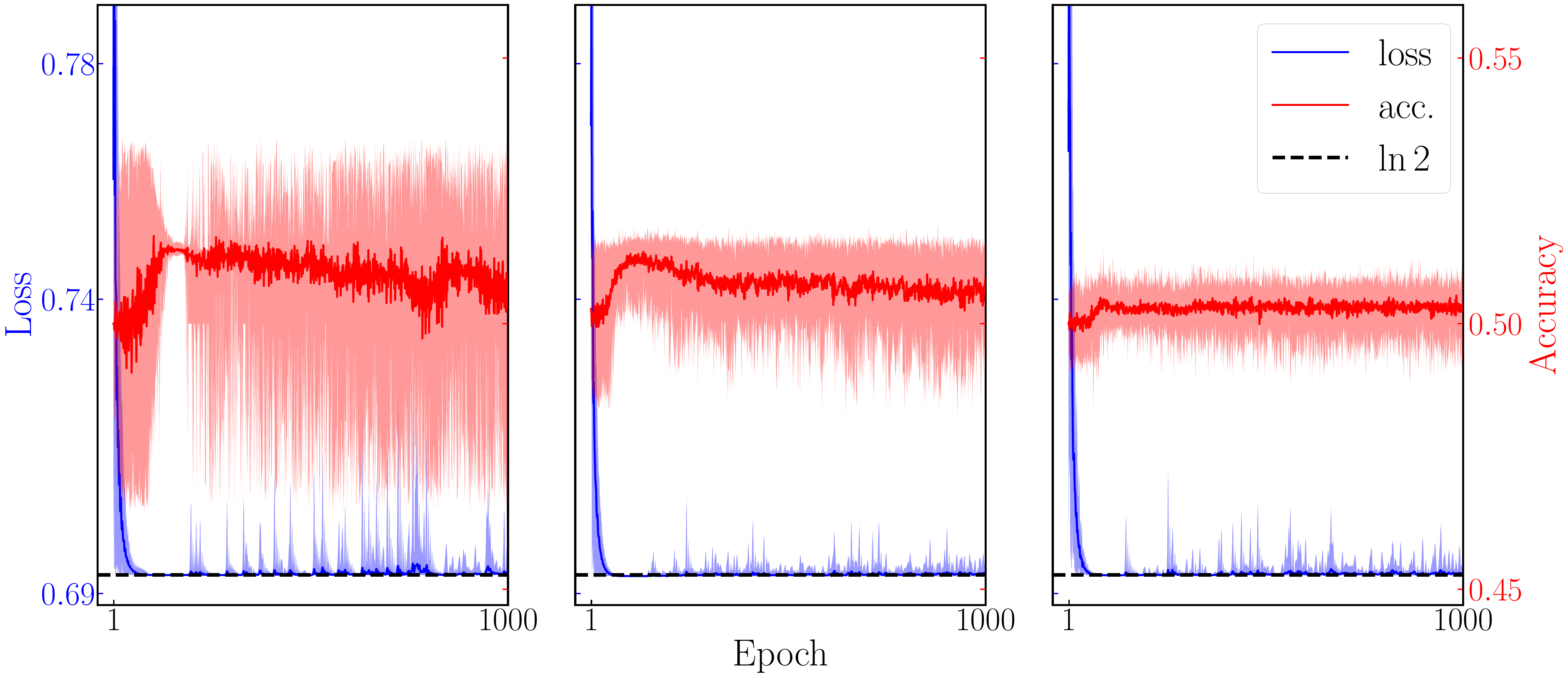}
	\caption{Performance of the quantum classifier trained by a 30-layer circuit. From left to right, the figures show the training loss and accuracy for the datasets in Fig.~\ref{fig:example_1}(a), Fig.~\ref{fig:example_2}(a), and Fig.~\ref{fig:example_3}(a), respectively. The solid blue and red line, respectively, represent the mean of the training loss and accuracy over 10 runs, and the shaded areas represent their respective ranges  across the 10 runs.}
	\label{fig:loss_30_layer}
\end{figure}

\section{QCNN circuit diagram}
\label{app:QCNN}
\setcounter{figure}{0}
\renewcommand{\thefigure}{D.\arabic{figure}}
The 10-qubit QCNN circuit  is shown in Fig.~\ref{fig:QCNN}. After amplitude encoding, the PQC consists of four convolutional layers (Conv) and pooling layers (Pool). In the pooling layers, one of the paired qubits is measured, and conditioned on the measurement result, a rotation is applied to the remaining qubit. The two-qubit gates are $XX(\theta) = e^{-i \frac{\theta}{2} (X \otimes X)}$, $YY(\theta) = e^{-i \frac{\theta}{2} (Y \otimes Y)}$, and $ZZ(\theta) = e^{-i \frac{\theta}{2} (Z \otimes Z)}$. The single-qubit gates are $R_x(\theta) = e^{-i \frac{\theta}{2} X}$, $ R_z(\theta) = e^{-i \frac{\theta}{2} Z}$, where $X$, $Y$, and $Z$ are Pauli operators. The parameters of the variational gates are independent.

\begin{figure}[htpb]
	\centering
	\includegraphics[width=0.64\textwidth]{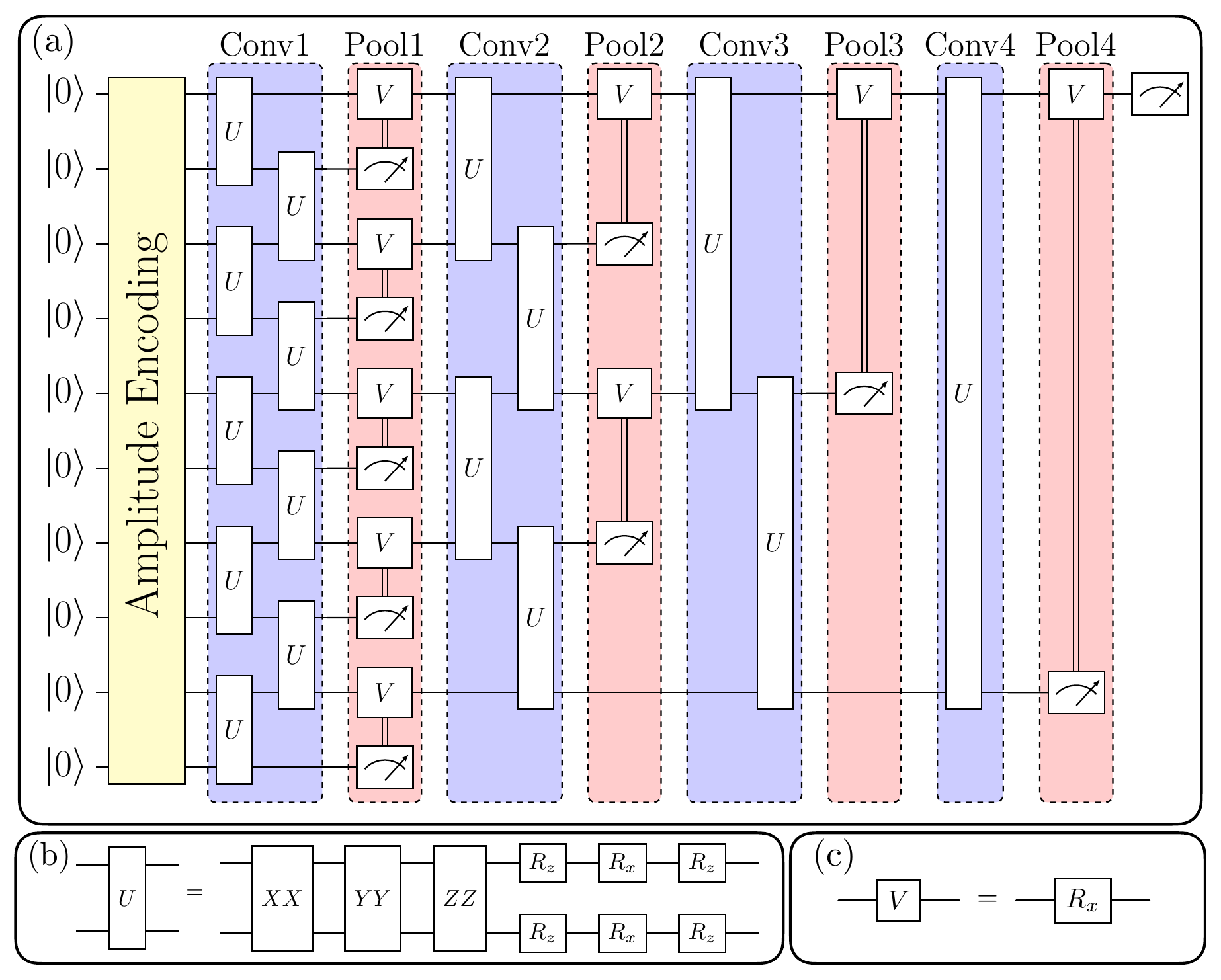}
	\caption{(a) The architecture of QCNN. (b) The circuit for gate $U$. (c) The circuit for gate $V$.}
	\label{fig:QCNN}
\end{figure}

\end{document}